\documentclass[graybox]{svmult}

\usepackage{mathptmx}       % selects Times Roman as basic font
\usepackage{helvet}         % selects Helvetica as sans-serif font
\usepackage{courier}        % selects Courier as typewriter font
\usepackage{type1cm}        % activate if the above 3 fonts are
\usepackage{makeidx}         % allows index generation
\usepackage{graphicx}        % standard LaTeX graphics tool
\usepackage{multicol}        % used for the two-column index
\usepackage[bottom]{footmisc}% places footnotes at page bottom

\usepackage{amsmath}
\usepackage{amsfonts}
\usepackage{amssymb}
\usepackage{bm}
\usepackage{subfigure} % for different captions of figures
\usepackage{float} % for positions of figures
\usepackage{color}
\usepackage{mathrsfs}
\usepackage{euscript}

\makeindex             % used for the subject index
\begin{document}

\title*{Multi-state asymmetric simple exclusion processes}
\author{Chihiro Matsui}
\institute{Chihiro Matsui \at Department of Mathematical Informatics, The University of Tokyo, 7-3-1 Hongo, Bunkyo-ku, Tokyo 113-8656, Japan, \email{matsui@stat.t.u-tokyo.ac.jp}
}

\maketitle

\abstract*{
It is known that the Markov matrix of the asymmetric simple exclusion process (ASEP) is invariant under the $U_q(sl_2)$ algebra. This is the result of the fact that the Markov matrix of the ASEP coincides with the generator of the Temperley-Lieb (TL) algebra, the dual algebra of the $U_q(sl_2)$ algebra. 
Various types of algebraic extensions have been considered for the ASEP. In this paper, we considered the multi-state extension of the ASEP, by allowing more than two particles to occupy the same site. We constructed the Markov matrix by dimensionally extending the TL generators and derived explicit forms of particle densities and currents on steady states. Then we showed how decay lengths differ from the original two-state ASEP under closed boundary conditions.
}

\abstract{
It is known that the Markov matrix of the asymmetric simple exclusion process (ASEP) is invariant under the $U_q(sl_2)$ algebra. This is the result of the fact that the Markov matrix of the ASEP coincides with the generator of the Temperley-Lieb (TL) algebra, the dual algebra of the $U_q(sl_2)$ algebra. 
Various types of algebraic extensions have been considered for the ASEP. In this paper, we considered the multi-state extension of the ASEP, by allowing more than two particles to occupy the same site. We constructed the Markov matrix by dimensionally extending the TL generators and derived explicit forms of particle densities and currents on steady states. Then we showed how decay lengths differ from the original two-state ASEP under closed boundary conditions.
}


%%%%%%%%%%%%%%%%%%%%%%%%%%%%%%%%%%%%%%%%%%%%%%%%%%%%%%%%%%%%%%%%%%%%%%%%%%%%%%%%%%%
%%%%%%%%%%%%%%%%%%%%%%%%%%%%%%%%%%%%%%%%%%%%%%%%%%%%%%%%%%%%%%%%%%%%%%%%%%%%%%%%%%%
\section{Introduction}
The asymmetric simple exclusion process (ASEP) is a one-dimensional exclusion process that describes discrete one-dimensional random walks. Among far-from-equilibrium systems, the ASEP is one of few examples which allows us the exact analysis. 
This model is first introduced in order to discuss dynamics of ribosome translocation along mRNA~\cite{bib:MGP68}. Then the model was applied in the context of traffic and transport systems~\cite{bib:CSS00}. 
As the exactly solvable non-equilibrium systems, the ASEP has been attracted for decades; After being used in the context of a diffusion process~\cite{bib:L85, bib:S91}, the connection to the solvable two-dimensional lattice system, the six-vertex model associated with the $U_q(sl_2)$ algebra, was pointed out in~\cite{bib:KDN90, bib:GS92}. 
Based on the $U_q(sl_2)$-invariance of the Markov matrix, the dynamical exponent was discussed besides the particle-density profiles and the particle currents were computed~\cite{bib:SS94}. 

The standard procedure to analyze the ASEP with general boundary conditions is the matrix product ansatz (MPA)~\cite{bib:DEHP93}. By assigning a matrix to each of the empty and occupied state, the steady state is given by the combination of those matrices which satisfy quadratic relations. 
Employing this method, various physical quantities such as particle-density profiles and particle currents, as well as steady states, were exactly computed. The existence of the phase transition in the particle densities with respect to incoming and outgoing hopping rates was obtained from the exact calculation~\cite{bib:K91, bib:DDM92, bib:DE93, bib:SD93}, as the evidence of strong dependence on boundary conditions of non-equilibrium systems. Matrices casted in a matrix product state satisfies quadratic relations which are understood in the context of the $q$-boson~\cite{bib:S94} and later whose connection with the Askey-Wilson polynomials was pointed out~\cite{bib:S99}. 

The Bethe ansatz method, first introduced to solve many body problems~\cite{bib:B31}, was also applied to the computation of physical quantities on the steady state of the ASEP. The periodic boundary condition~\cite{bib:GS92*}, the system with infinite length~\cite{bib:K95, bib:S97}, and the open boundary case~\cite{bib:GE06} were discussed in each literature. The relation to the Kardar-Parisi-Zhang (KPZ) universality classes was also discussed through the dynamical exponent derived from the analysis of the Bethe-ansatz equations for the periodic boundary case~\cite{bib:GS92*} and through the asymptotic behavior of correlation functions for the closed boundary case~\cite{bib:SS94}. 
Even though the connection between the MPA and the Bethe vectors have not been clear for years, this question was unveiled in~\cite{bib:GM06}, by expressing creation and annihilation operators of the algebraic Bethe ansatz in terms of the operators in the matrix product states.

The model was algebraically extended to be associated with various algebras. The possible realizations of the Hecke algebra were pointed out which describe the time evolution of reaction-diffusion processes~\cite{bib:AR93}. 
Among various algebraic extensions, the ASEP with multi-species has been closely studied~\cite{bib:ADR98, bib:K99, bib:FJ99, bib:FM07, bib:AKSS09, bib:EFM09, bib:PEM09, bib:AAMP11, bib:AM13}. The steady state was constructed by the MPA, whose matrix elements satisfy similar quadratic relations as in the one-species case~\cite{bib:EFM09}. The spectrum of the Markov matrix was also studied on the periodic system, in which the dynamical exponent was found to be the same as the one-species case~\cite{bib:AKSS09}. 
It has been also pointed out that various models are mapped onto the ASEP. The interesting example is the zero-range process~\cite{bib:E97}, which often appears as another traffic model with hopping rates depending on the number of particles in each box.

In this paper, another algebraic extension of the ASEP, the ASEP associated with the higher-dimensional representation of the $U_q(sl_2)$ algebra is proposed. The original two-state ASEP is characterized by the two-dimensional basis associated with the fundamental representations of the $U_q(sl_2)$ algebra defined by the empty state and the one-particle state. The idea is based on the dimensional extension of the Markov matrix, which results in the multi-state extension of the ASEP. Although it is not straightforward how to extend the Markov matrix itself to the higher-dimensional one by keeping integrability of the model, we first considered the dimensional extension to the associated TL generator and then deformed the extended TL generators to be a Markov matrix. 
The dimensional extension of the TL generator was proposed in~\cite{bib:Z07}, by using the notion of fusion. Applying the projection operator, $\ell$-fold tensor product of two-dimensional vector spaces are $q$-symmetrized, which allows us to take out the $(\ell+1)$-dimensional representation, the largest irreducible subspace. 
In general, the fused TL generators satisfy neither the probability conservation nor the positivity of probability. The probability conservation is satisfied if a matrix has zero-sum columns. We have found the similarity transformation for the fused TL generators to have zero-sum columns. 
In order to ensure the positivity of probability, we focused on the fact that $\ell$ different types of TL generators are constructed by the fusion procedure. Thus, even if each of the generators does not satisfy the positivity condition, which is given by the positiveness of off-diagonal matrix elements, there would be a linear combination of different types of TL generators which makes all the off-diagonal elements positive. We showed the existence of the parameter regimes for the positivity condition and conjectured the restrictions among the coefficients for the linearly combined TL generators to satisfy the positivity condition. 
In order to obtain characterizing feature of the multi-state ASEPs, two special limits, the symmetric case and the totally asymmetric case, were discussed. In these limits, the physical meanings of the coefficients of combined generators become clear. 
The virtue of employing the fusion procedure is that, the TL generator is dimensionally extended by keeping the commutativity with the $U_q(sl_2)$ generators. Consequently, the Markov matrices associated with higher-dimensional TL generators, which describe multi-state stochastic processes, keep the $U_q(sl_2)$-invariance. 
This construction of the multi-state ASEPs suggests the existence of a new family of integrable stochastic models associated with the higher-dimensional representations of quantum groups. 

It is an interesting question how the number of states in the model affects on physical behaviors. For example, we would expect different exponents in decay lengths for the ASEPs with different number of states. Due to the possible variety of incoming and outgoing hopping rates, the rich structure of phases in the particle-density profiles is expected under general boundary conditions. 
This paper focuses on phenomena in the steady state under the closed boundary conditions, in which the whole system is invariant under the $U_q(sl_2)$ algebra. We first derive the steady states of the multi-state ASEPs, which are given by the basis of the $U_q(sl_2)$ algebra. Although computations are cumbersome if one works with high-dimensional representations, we show how calculations can be proceeded on the two-dimensional representations by means of the projection operators. 
We derive the explicit forms of particle-density profiles, which show the transition from the zero-particle domain to the high-density domain. Then, we prove that the interval of the transition domain depends on the number of states of the system. Consequently, we also show that the existence of right-moving currents and left-moving currents, although they compensate with each other, at the transition domain from the zero density to the high density.

This paper is organized as follows. In Section~2, we review the algebraic aspects of the ASEP. The minimum necessary facts about the $U_q(sl_2)$ algebra are also given for the construction of multi-state ASEPs and the calculation of particle densities and currents. 
Section~3 is devoted to the model settings of the multi-state ASEPs, including the construction of the higher-dimensional TL generators. We give detailed explanation about how the fused TL generators are deformed to satisfy the probability conservation and the positivity of probability. The symmetric limit and totally-asymmetric limit of the multi-state ASEPs are discussed to give physical meanings to the coefficients of a linear combination of the fused TL generators. 
In Section~4, we show the exact calculations of particle-density profiles and  currents as examples of physical quantities that can be computed on the multi-state ASEPs. 
In Section~5, the large-volume limit is considered. The asymptotic behaviors of density profiles are carefully analyzed and the density transition from the low-density phase to the high-density phase was obtained. It was also found that the decay length of the density transition depends on the number of states of the process. 
The last section is devoted to concluding remarks and open problems. 

%%%%%%%%%%%%%%%%%%%%%%%%%%%%%%%%%%%%%%%%%%%%%%%%%%%%%%%%%%%%%%%%%%%%%%%%%%%%%%%%%%%
%%%%%%%%%%%%%%%%%%%%%%%%%%%%%%%%%%%%%%%%%%%%%%%%%%%%%%%%%%%%%%%%%%%%%%%%%%%%%%%%%%%
\section{A brief review of ASEP}

The asymmetric simple exclusion process (ASEP) is a stochastic process defined on a one-dimensional lattice consisting of $N$ sites with a variable $\tau_i \in \{0, 1\}$ attached to each site $i$. This variable $\tau_i$ is, in a physical sense, considered as the number of particles admitted in the $i$th box. The transition rules is determined by the local transition rates defined for the configuration of variables on two neighboring sites $(\tau_i, \tau_{i+1})$; the transition $(1,0) \to (0,1)$ occurs with a rate $p_R$, while the transition $(0, 1) \to (1,0)$ with a rate $p_L$. 

The time evolution of the configuration of variables is given by the differential-difference equation. By writing the state of the whole system as 
\begin{equation} \label{state-def1}
 |\tau_1, \dots, \tau_N \rangle := |\tau_1 \rangle \otimes \cdots \otimes |\tau_N \rangle,  
\end{equation} 
where $|\tau_i \rangle \in \mathbb{C}^2$, the vector of configurations at time $t$ is expressed as 
\begin{equation} \label{prob_vec}
 |P(t) \rangle = \sum_{\tau_i \in \{0, 1\}} p(t; \tau_1, \dots, \tau_N) |\tau_1, \dots, \tau_N \rangle 
\end{equation}
with the probabilities $p(t; \tau_1, \dots, \tau_N)$ to obtain each configuration. Then the time evolution of $|P(t) \rangle$ is simply given by the differential-difference equation of $p(t; \tau_1, \dots, \tau_N)$: 
\begin{equation} \label{time-evol}
\begin{split}
 \frac{d}{dt} p(t; \tau_1, \dots, \tau_N) = &\sum_{i=1}^N \Theta (\tau_{i+1} - \tau_i)\, p(t; \tau_1, \dots, \tau_{i+1}, \tau_i, \dots, \tau_N) \\
 &- \sum_{i=1}^N \Theta(\tau_i - \tau_{i+1})\, p(t; \tau_1, \dots, \tau_i, \tau_{i+1}, \dots, \tau_N), 
\end{split}
\end{equation}
where 
\begin{equation}
 \Theta(x) = 
  \begin{cases}
   -p_R & x<0 \\
   0 & x=0 \\
   p_L & x>0. 
   \end{cases}
\end{equation}

In the physics realm, the time evolution of $|P(t) \rangle$ (\ref{time-evol}) is often written in a matrix form, which is known as the master equation: 
\begin{equation}
 \frac{d}{dt} |P(t) \rangle = M |P(t) \rangle, 
\end{equation}
where $M$ is the Markov matrix obtained from the update rules (\ref{time-evol}): 
\begin{align}
%\begin{split}
&M = \sum_{i=1}^{N-1} M_{i,i+1}, \\
&M_{i,i+1} = \bm{1}_1\otimes \bm{1}_2\otimes \cdots \bm{1}_{i-1}\otimes
\begin{pmatrix}
0 & 0 & 0 & 0\\
0 & -p_L & p_R & 0\\
0 & p_L & -p_R & 0\\
0 & 0 & 0 & 0
\end{pmatrix}_{i,i+1}
\otimes \bm{1}_{i+2}\otimes \cdots \otimes \bm{1}_N, \label{update}
%\end{split}
\end{align}
where the matrix $M_{i,i+1}$ nontrivially acts only on the $i$th and $(i+1)$th spaces of the $N$-fold tensor product of the fundamental representations. 
The equation (\ref{update}) indeed describes the process of particle-hopping to the right with a rate $p_R$ and to the left with a rate $p_L$ under the choice of the one-site state $|\tau_i \rangle$ as  
\begin{equation} \label{2state-basis}
 |0 \rangle_{i} = \begin{pmatrix} 1 \\ 0 \end{pmatrix}_{i}, 
\qquad
 |1 \rangle_{i} = \begin{pmatrix} 0 \\ 1 \end{pmatrix}_{i}. 
\end{equation}
Let us remark that under the closed boundary conditions, no incoming or outgoing particle is obtained on the system. 

%%%%%%%%%%%%%%%%%%%%%%%%%%%%%%%%%%%%%%%%%%%%%%%%%%%%%%%%%%%%%%%%%%%%%%%%%%%%%%%
\subsection{$U_q(sl_2)$-invariance of the Markov matrix}
The ASEP is an integrable stochastic process since the update operator $M_{i,i+1}$ is identified with a Temperley-Lieb (TL) generator: 
\begin{equation}
e_i = 
\begin{pmatrix}
0 & 0 & 0 & 0\\
0 & q^{-1} & -1 & 0 \\
0 & -1 & q & 0\\
0 & 0 & 0 & 0
\end{pmatrix}_{i,i+1}, 
\end{equation}
\ifx10
\begin{align}
& P_{i,i+1} M_{i,i+1} = -M_{i,i+1}\\
& M_{i,i+1}^2 = -M_{i,i+1}\\
& M_{i,i+1} M_{i+1,i+2} M_{i,i+1} = pq M_{i,i+1}\\
& M_{i,i+1} M_{i-1,i} M_{i,i+1} = pq M_{i,i+1}\\
& M_{i,i+1} M_{j,j+1} = M_{j,j+1} M_{i,i+1} \qquad |i-j| \ge 2, 
\end{align}
\fi 
which satisfies the following algebraic relations (Fig.~\ref{fig:TLrel}): 
\begin{equation} \label{TL-relations}
\begin{split}
& e_i^2 = (q+ q^{-1}) e_i, \\
& e_i e_{i+1} e_i = e_i, \\
& e_i e_j = e_j e_i, \qquad |i-j| \ge 2. 
\end{split}
\end{equation}
Applying the following similarity transformation:  
\begin{equation} \label{TL-ASEP}
 U = \otimes_{i=1}^N U_i= \otimes_{i = 1}^N 
	\begin{pmatrix}
		1 & 0 \\
		0 & q^{i-1}
	\end{pmatrix}_i, 
	\qquad 
	q = \sqrt{\frac{p_R}{p_L}} > 0 
\end{equation}
leads to the update operator (\ref{update}) written by the TL generator; 
the update operator is related to the TL generator via 
\begin{equation}
  M_{i,i+1} = -\sqrt{p_R p_L}\, U_{i,i+1} e_i U_{i,i+1}^{-1}, 
\end{equation}
for which one can easily check that the relations (\ref{TL-relations}) hold up to overall factors. 

Let us introduce the spin operators given by
\begin{equation}
 S^+ = \begin{pmatrix} 0 & 1 \\ 0 & 0 \end{pmatrix}, \quad
 S^- = \begin{pmatrix} 0 & 0 \\ 1 & 0 \end{pmatrix}, \quad
 q^{S^z} = \begin{pmatrix} q^{\frac{1}{2}} & 0 \\ 0 & q^{-\frac{1}{2}} \end{pmatrix}. 
\end{equation}
The spin operators are known to generate the $U_q(sl_2)$ algebra. 
The TL algebra is a dual algebra of the $U_q(sl_2)$ algebra in the sense that the TL generators commute with those of the $U_q(sl_2)$ algebra: 
\begin{equation}
 [e_i,\, \Delta(X)] = 0, 
  \qquad 
  X \in \{S^{\pm}, q^{S^z}\}, 
\end{equation}
where $\Delta$ represents the coproduct defined in Appendix~1. 
The commutativity of these generators allows us to compute the exact steady state of the ASEP, and consequently, the exact physical quantities. 

In our notations (\ref{2state-basis}), the empty state $|0 \rangle$, which is naturally invariant under the time development, is identified with the highest weight vector of the $U_q(sl_2)$ algebra. Namely, the steady state of the ASEP with no particle is given by the highest weight vector of the $U_q(sl_2)$ algebra: 
\begin{equation}
 \frac{d}{dt} |0 \rangle = M |0 \rangle = 0. 
\end{equation}
Taking into account that the $U_q(sl_2)$ generators commute with the TL generators, and consequently, with the Markov matrix up to the similarity transformation, we have 
\begin{equation}
 M U (\Delta^{(N)}(S^-))^n U^{-1} |0 \rangle
  = U (\Delta^{(N)}(S^-))^n U^{-1} M |0 \rangle
  = 0, 
\end{equation}
where $\Delta^{(N)}$ is the $N$th coproduct (Appendix~1). 
Thus, a series of steady states is obtained as the vector basis of the $U_q(sl_2)$ algebra constructed by applying the generator $S^-$ to the highest weight vector.

%%%%%%%%%%%%%%%%%%%%%%%%%%%%%%%%%%%%%%%%%%%%%%%%%%%%%%%%%%%%%%%%%%%%%%%%%%%%%%%%%%%%%%%%%%%%
%%%%%%%%%%%%%%%%%%%%%%%%%%%%%%%%%%%%%%%%%%%%%%%%%%%%%%%%%%%%%%%%%%%%%%%%%%%%%%%%%%%%%%%%%%%%
\section{Multi-state ASEP}
As we reviewed in the previous section, the update operators of the ASEP satisfy the TL relations. Based on this fact, the extension of the integrable stochastic process to multi-state cases is, taking into account that a state of each box is given by (\ref{2state-basis}), achieved by constructing the higher-dimensional update matrices which still commute with the $U_q(sl_2)$ generators. 
The dimensional extension of the TL generators have been discussed by P. Zinn-Justin in~\cite{bib:Z07}, in which they have constructed fused TL generators from $\ell$-fold tensor products of the fundamental representations with the use of a projection operator. Although the fused TL generators themselves cannot be update operators, as they enjoy neither the probability conservation nor the positivity of probability in nature, we found similarity transformations, which make the fused TL generators satisfy the probability conservation. Moreover, we show that a proper linear combination of different types of the fused TL generators has only positive off-diagonal matrix elements, which ensures the positivity of probability.  

%%%%%%%%%%%%%%%%%%%%%%%%%%%%%%%%%%%%%%%%%%%%%%%%%%%%%%%%%%%%%%%%%%%%%%%%
\subsection{Higher-spin TL generators}
The TL generator of a higher-dimensional representation is constructed from tensor products of fundamental representations. Since the fundamental representation of the TL generator is the two-dimensional representation, we consider the $\ell$-fold tensor products to construct the $(\ell+1)$-dimensional TL generators, which are associated with the $(\ell+1)$-state ASEP. 
Let us consider the simplest example, {\it i.e.} the construction of three-dimensional TL generators. Following the above statement, we consider the two-fold tensor product of the fundamental representations. As in the spin composition, the irreducible decomposition leads us to obtain the totally $q$-symmetric subspace, which has the three-dimensional representation, and the totally $q$-asymmetric subspace with the one-dimensional representation. Then, we take out the totally $q$-symmetric subspace by using the projection operator introduced later. 
Similarly, the largest irreducible subspace of the $\ell$-fold tensor product of the fundamental representations is totally $q$-symmetric, which has the $(\ell+1)$-dimensional representation, and thus can be taken out by the projection operator. 

The projection operator is recursively constructed from the TL generators: 
\begin{equation}
Y^{(k+1)}(e_j,\dots,e_{j+k-1}) = Y^{(k)}(e_j,\dots,e_{j+k-2}) \left(1 - \frac{U_{k-1}(\tau)}{U_{k}(\tau)} e_{k+\ell-1} \right) Y^{(k)}(e_j,\dots,e_{j+k-2}), 
\end{equation}
with the initial condition $Y^{(1)} = 1$. The functions $U_k(\tau)$ are the Chebyshev polynomials of the second kind with a parameter $\tau$ given by $\tau=(q+q^{-1})/2$. The superscripts of the projection operator denote how many spaces the operator acts on, that is, $Y^{(\ell)}$ takes out the largest irreducible subspace from $\ell$-fold tensor product of the fundamental representations. For instance, the projection operator $Y^{(2)}$, which acts on the $k$th and $(k+1)$th spaces, is obtained as 
\begin{equation}
 Y^{(2)}(e_i) = Y^{(1)} 
 \left(
 1 - \frac{U_0(\tau)}{U_1(\tau)} e_i
 \right)
 Y^{(1)}
 = \begin{pmatrix}
     1 & 0 & 0 & 0 \\
     0 & \frac{q}{q+q^{-1}} & \frac{1}{q+q^{-1}} & 0 \\
     0 & \frac{1}{q+q^{-1}} & \frac{q^{-1}}{q+q^{-1}} & 0 \\
     0 & 0 & 0 & 1
    \end{pmatrix}_{i,i+1}. 
\end{equation}
Using this projection operator, there are two possible ways to fuse the TL operator. From now on, we simply write $Y^{(k)}(e_j,\dots,e_{j+k-2})$ by $Y^{(k)}_j$ and introduce a new notation $Y^{(k)} = \prod_{j=1}^N Y^{(k)}_j$. Then the three-dimensional fused TL generators are given by 
\begin{align}
 &e_i^{(2;1)} = Y^{(2)} e_{2(i-1) + 2} Y^{(2)}, \label{TL-3dim1} \\
 &e_i^{(2;2)} = Y^{(2)} e_{2(i-1) + 2} e_{2(i-1) + 1} e_{2(i-1) + 3} e_{2(i-1) + 2} Y^{(2)}. \label{TL-3dim2}
\end{align}
The graphical representations of $e_i^{(2;1)}$ and $e_i^{(2;2)}$ are given in Fig.~\ref{fig:TL-3dim} (Appendix~2). These operators are known to satisfy the $SO(3)$ Birman-Murakami-Wenzl (BMW) algebra \cite{bib:M87, bib:BW89, bib:FR02} given in Appendix~3. 

The fused TL generator of an $e_i^{(\ell;r)}$ is given by 
\begin{equation} \label{TL-const}
e^{(\ell;r)}_i = Y^{(\ell)}\, e_{\ell i}\, e_{\ell i-1}\, e_{\ell i+1} \cdots e_{\ell i-r+1} \cdots e_{\ell i+r-1} \cdots e_{\ell i-1}\, e_{\ell i+1}\, e_{\ell i}\, Y^{(\ell)}_{\ell i}, 
\end{equation}
where $r=1,2,\dots,\ell$ indicates the type of the TL generators (Fig.~\ref{fig:TL}). In general, $\ell$ different kinds of the TL generators can be constructed by fusing $\ell$ fundamental representations. 
It is important to note that each of $\ell$ kinds of the TL generators constructed in this way commutes with the $U_q(sl_2)$ generators of the $(\ell+1)$-dimensional representations: 
\begin{equation} \label{TL-Uq}
 [e_i^{(\ell;r)},\, \Delta(X)] = 0, 
  \qquad
  X \in \{S^{\pm}, q^{S^z}\}. 
\end{equation}
\begin{figure}
\begin{center}
 \includegraphics[scale=0.65]{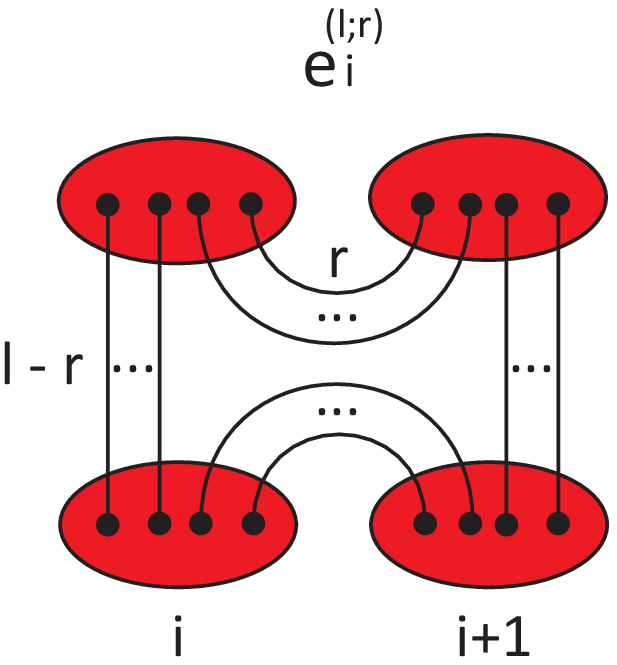}
 \caption{The TL generator of the $(\ell+1)$-dimensional representation. $\ell$ spaces are fused by the projection operator $Y^{(\ell)}$, denoted by a red circle. The operator $e_i^{(\ell;r)}$ includes $r$ arches. } \label{fig:TL}
\end{center}
\end{figure}

Due to the existence of a parameter $q$, the projection operator $Y_k^{(\ell)}$ does not simply symmetrize $\ell$ spaces, as in the $SU(2)$ case, but $q$-symmetrizes them. For instance, the following three vectors are $q$-symmetric under the transpose of two spaces: 
\begin{equation}
 \begin{pmatrix} 1 \\ 0 \end{pmatrix} 
 \otimes 
  \begin{pmatrix} 1 \\ 0 \end{pmatrix},
  \quad
   q^{1/2} \begin{pmatrix} 1 \\ 0 \end{pmatrix} 
 \otimes 
  \begin{pmatrix} 0 \\ 1 \end{pmatrix}
  +
   q^{-1/2} \begin{pmatrix} 0 \\ 1 \end{pmatrix} 
 \otimes 
  \begin{pmatrix} 1 \\ 0 \end{pmatrix},
  \quad
\begin{pmatrix} 0 \\ 1 \end{pmatrix} 
 \otimes 
  \begin{pmatrix} 0 \\ 1 \end{pmatrix}. 
\end{equation}
One obtains later in (\ref{np-vector}) that these vectors are the zero-, one-, and two-particle local states of the three-state ASEP, respectively.

%%%%%%%%%%%%%%%%%%%%%%%%%%%%%%%%%%%%%%%%%%%%%%%%%%%%%%%%%%%%%%%%%%%%%%%%
\subsection{Update operators of multi-state ASEPs}
Now we construct update operators of the multi-state ASEPs. In order to be an update operator, the following two conditions should be satisfied: 
\begin{enumerate}
 \item[(i)] The sum of each column should be zero ({\it the principle of probability preservation}). 
 \item[(ii)] The off-diagonal elements should be negative values, while the diagonal elements should be positive values ({\it positivity of probability}). 
\end{enumerate}

\subsubsection{Probability conservation}
As the simplest example, we first show how the three-dimensional TL generator of the type-1 $e^{(2;1)}_i$ is modified to satisfy the probability conservation. The operator is given by the $9$-by-$9$ matrix since it acts on a two-fold tensor product of three-dimensional vector spaces: 
\begin{equation}
 e^{(2;1)}_i = 
\left(
\begin{smallmatrix}
 0 & 0 & 0 & 0 & 0 & 0 & 0 & 0 & 0 \\
 0 & \frac{q^{-2}}{q+q^{-1}} & 0 & -\frac{1}{q+q^{-1}} & 0 & 0 & 0 & 0 & 0 \\
 0 & 0 & \frac{1}{q} & 0 & -\frac{q}{\left(q+q^{-1}\right)^2} & 0 & 0 & 0 & 0 \\
 0 & -\frac{1}{q+q^{-1}} & 0 & \frac{q^2}{q+q^{-1}} & 0 & 0 & 0 & 0 & 0 \\
 0 & 0 & -q & 0 & \frac{q^2-1+q^{-2}}{q+q^{-1}} & 0 & -\frac{1}{q} & 0 & 0 \\
 0 & 0 & 0 & 0 & 0 & \frac{q^{-2}}{q+q^{-1}} & 0 & -\frac{1}{q+q^{-1}} & 0 \\
 0 & 0 & 0 & 0 & -\frac{q^{-1}}{\left(q+q^{-1}\right)^2} & 0 & q & 0 & 0 \\
 0 & 0 & 0 & 0 & 0 & -\frac{1}{q+q^{-1}} & 0 & \frac{q^2}{q+q^{-1}} & 0 \\
 0 & 0 & 0 & 0 & 0 & 0 & 0 & 0 & 0 
\end{smallmatrix}
\right)_{i,i+1}. 
\end{equation}
This matrix apparently does not satisfy the probability conservation and, in general, the same is true for the other TL fused generators. However, we found for any dimensional representation the existence of a similarity transformation which makes the TL generators satisfy the probability preservation. 
\begin{proposition}
The TL generators of the $(\ell+1)$-dimensional representation satisfy the principle of probability conservation after a similarity transformation given by the following matrix: 
\begin{align} \label{U-mat}
 U^{(\ell)} = \otimes_{i=1}^N U_i = \otimes_{i = 1}^N
\left(
\begin{smallmatrix}
	a_0^{(i)} & & & & & \\
	& a_1^{(i)} & & & & \\
	& & \ddots & & & \\
	& & & a_k^{(i)} & & \\
	& & & & \ddots & \\
	& & & & & a_\ell^{(i)}
\end{smallmatrix}
\right)_i 
\end{align}
whose matrix elements are given by $a_k^{(i)} = q^{k \ell (i-1)}$. 
\end{proposition}
\begin{proof}
The matrix elements of $e_i^{(\ell;k)}$ is obtained by calculating $\langle r-x| \otimes \langle s+x|\; e_i^{(\ell;k)} \;|r \rangle \otimes |s \rangle$. After transformed by $U^{(\ell)}$, the matrix elements are transformed as $q^{\ell x} \langle r-x| \otimes \langle s+x| e_i^{(\ell;k)} |r \rangle \otimes |s \rangle$. This implies that the deformation of the TL generator does not depend on the site number $i$. 
Using the explicit forms of the fused TL generators (\ref{TL-ldim}) in Appendix~4, one obtains zero-sum columns after the transformation by $U^{(\ell)}$ from direct calculation. 
\end{proof}
The important property is that the matrix $U^{(\ell)}$ simultaneously transforms any kinds of the TL generators, as long as they belong to the same dimensional representation, in order to satisfy the probability conservation. As an illustration, let us give the explicit forms of the transformed three-dimensional TL generators: 
\begin{equation} \label{e21-def}
 U_{i,i+1}^{(2)} e_i^{(2;1)} (U_{i,i+1}^{(2)})^{-1}
 = 
\left(
\begin{smallmatrix}
 0 & 0 & 0 & 0 & 0 & 0 & 0 & 0 & 0 \\
 0 & \frac{q^{-2}}{q+q^{-1}} & 0 & -\frac{q^2}{q+q^{-1}} & 0 & 0 & 0 & 0 & 0 \\
 0 & 0 & \frac{1}{q} & 0 & -\frac{q^3}{\left(q+q^{-1}\right)^2} & 0 & 0 & 0 & 0 \\
 0 & -\frac{q^{-2}}{q+q^{-1}} & 0 & \frac{q^2}{q+q^{-1}} & 0 & 0 & 0 & 0 & 0 \\
 0 & 0 & -\frac{1}{q} & 0 & \frac{q^2-1+q^{-2}}{q+q^{-1}} & 0 & -q & 0 & 0 \\
 0 & 0 & 0 & 0 & 0 & \frac{q^{-2}}{q+q^{-1}} & 0 & -\frac{q^2}{q+q^{-1}} & 0 \\
 0 & 0 & 0 & 0 & -\frac{q^{-3}}{\left(q+q^{-1}\right)^2} & 0 & q & 0 & 0 \\
 0 & 0 & 0 & 0 & 0 & -\frac{q^{-2}}{q+q^{-1}} & 0 & \frac{q^2}{q+q^{-1}} & 0 \\
 0 & 0 & 0 & 0 & 0 & 0 & 0 & 0 & 0 
\end{smallmatrix}
\right)_{i,i+1}, 
\end{equation}
\begin{equation} \label{e22-def}
 U_{i,i+1}^{(2)} e_i^{(2;2)} (U_{i,i+1}^{(2)})^{-1}
 =
\left(
\begin{smallmatrix}
 0 & 0 & 0 & 0 & 0 & 0 & 0 & 0 & 0 \\
 0 & 0 & 0 & 0 & 0 & 0 & 0 & 0 & 0 \\
 0 & 0 & \frac{1}{q^2} & 0 & -\frac{q}{q+q^{-1}} & 0 & q^4 & 0 & 0 \\
 0 & 0 & 0 & 0 & 0 & 0 & 0 & 0 & 0 \\
 0 & 0 & -q^{-3}(q+q^{-1}) & 0 & 1 & 0 & -q^3 \left(q+q^{-1}\right) & 0 & 0 \\
 0 & 0 & 0 & 0 & 0 & 0 & 0 & 0 & 0 \\
 0 & 0 & \frac{1}{q^4} & 0 & -\frac{q^{-1}}{q+q^{-1}} & 0 & q^2 & 0 & 0 \\
 0 & 0 & 0 & 0 & 0 & 0 & 0 & 0 & 0 \\
 0 & 0 & 0 & 0 & 0 & 0 & 0 & 0 & 0 
\end{smallmatrix} 
\right)_{i,i+1}, 
\end{equation}
both of which satisfy the probability conservation. 
\subsubsection{Positivity of probability}
Besides the principle of probability preservation, the update matrix should satisfy the positivity of probability which is realized by such conditions that the off-diagonal elements of the update matrix are positive values and the diagonal elements are negative. 
For instance, the transformed operator $M_{i,i+1}^{(2;1)} = -U_{i,i+1}^{(2)} e^{(2;1)}_i (U_{i,i+1}^{(2)})^{-1}$ in nature satisfies the posivitivity of probability, as is obtained from (\ref{e21-def}), since we have $q \geq 0$ (\ref{TL-ASEP}). However, the operator $M_{i,i+1}^{(2;2)} = -U_{i,i+1}^{(2)} e^{(2;2)}_i (U_{i,i+1}^{(2)})^{-1}$ (\ref{e22-def}) has negative off-diagonal elements at the $(3,7)$ and $(7,3)$-elements. 
\begin{proposition}
Consider the following linear combination of the fused TL generators: 
\begin{equation}
 M^{(2)}_{i,i+1} = b_1^{(2)} M_{i,i+1}^{(2;1)} + b_2^{(2)} M_{i,i+1}^{(2;2)}. 
\end{equation}
The matrix $M_{i,i+1}^{(2)}$ satisfies the positivity condition as long as $\beta = b_2^{(2)} / b_1^{(2)}$ satisfies the following conditions: 
\begin{equation} \label{cond_comb}
\begin{split}
		&-\frac{q^2}{q+q^{-1}} < \beta < 0 \qquad (0 < q \leq 1) \\
		&-\frac{q^{-2}}{q+q^{-1}} < \beta < 0 \qquad (1 \leq q). 
\end{split}
\end{equation}
\end{proposition}
\begin{proof}
 We prove the above proposition from the direct calculation of the matrix elements. The update matrix $M_{i,i+1}^{(2)}$ is given by 
\begin{equation} \label{3state-markov}
 M_{i,i+1}^{(2)} = 
\left(
 \begin{smallmatrix}
 0 & 0 & 0 & 0 & 0 & 0 & 0 & 0 & 0 \\
 0 & -\frac{1}{q^2(q+q^{-1})} & 0 & \frac{q^2}{q+q^{-1}} & 0 & 0 & 0 & 0 & 0 \\
 0 & 0 & -\frac{q+\beta}{q^2} & 0 & \frac{q^3 + \beta q(q+q^{-1})}{(q+q^{-1})^2} & 0 & -q^4 \beta & 0 & 0 \\
 0 & \frac{1}{q^2 (q+q^{-1})} & 0 & -\frac{q^2}{q+q^{-1}} & 0 & 0 & 0 & 0 & 0 \\
 0 & 0 & \frac{q^3 + \beta q(q+q^{-1})}{q^4} & 0 & -\frac{q^4 - q^2 + 1 + \beta q^2 (q+q^{-1})}{q^2 (q+q^{-1})} & 0 & q \left(\beta q^3+\beta q+1\right) & 0 & 0 \\
 0 & 0 & 0 & 0 & 0 & -\frac{1}{q^2 (q+q^{-1})} & 0 & \frac{q^2}{q+q^{-1}} & 0 \\
 0 & 0 & -\frac{\beta}{q^4} & 0 & \frac{\beta q^2 (q+q^{-1}) + 1}{q^3 (q+q^{-1})^2} & 0 & -q (q \beta+1) & 0 & 0 \\
 0 & 0 & 0 & 0 & 0 & \frac{1}{q^2 (q+q^{-1})} & 0 & -\frac{q^2}{q+q^{-1}} & 0 \\
 0 & 0 & 0 & 0 & 0 & 0 & 0 & 0 & 0 
 \end{smallmatrix}
\right)_{i,i+1}. 
\end{equation}
In order to make all off-diagonal elements positive at the same time with making all diagonal elements negative values, $\beta$ needs to satisfy the following six conditions: 
\begin{equation} \label{ineq}
\begin{split}
 &q (q^3 + \beta + q^2 \beta) > 0, \hspace{25mm}
  -q\beta > 0, \\
 &1 + q\beta + q^3\beta > 0, \hspace{31mm}
 q (q + \beta) > 0, \\
 &1 - q^2 + q^4 + q\beta + q^3\beta > 0, \hspace{15mm}
 1 + q\beta > 0. 
\end{split}
\end{equation}
Taking into account that our $q$ takes only positive values (\ref{TL-ASEP}), the inequalities (\ref{ineq}) can be solved as 
\begin{equation}
\begin{split}
 &\beta < 0, \\
 &\beta > {\rm max.}\left\{ -\tfrac{q^2}{q+q^{-1}},\, -\tfrac{q^{-2}}{q+q^{-1}},\, -q,\, -q^{-1},\, -\tfrac{q^2-1+q^{-2}}{q+q^{-1}} \right\}. 
\end{split}
\end{equation}
From Fig.~\ref{fig:multi-markov}, we finally obtain the condition (\ref{cond_comb}) for $\beta$. 
\end{proof}
\begin{figure}
\begin{center}
 \includegraphics[scale=0.75]{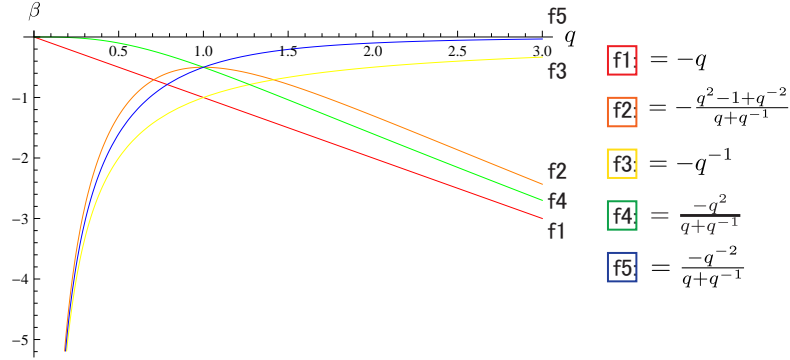}
 \caption{Behaviors of each function in (\ref{ineq}) with respect to the parameter $q$. Since the green line crosses with the blue line at $q = 1$, we have $-q^{2}/(q+q^{-1}) < \beta < 0$ for $0 \leq q \leq 1$ and $q^{-2}/(q+q^{-1}) < \beta < 0$ for $1\leq q$. } \label{fig:multi-markov}
\end{center}
\end{figure}

For the multi-state ASEP with arbitrary $\ell$, we consider the following combination of the fused TL generators: 
\begin{equation} \label{comb-l}
 M_{i,i+1}^{(\ell)} = \sum_{r=1}^{\ell} b_r^{(\ell)} M_{i,i+1}^{(\ell;r)}, 
\end{equation}
where $M_{i,i+1}^{(\ell;r)} = -U_{i,i+1}^{(\ell)} e_i^{(\ell;r)} (U_{i,i+1}^{(\ell)})^{-1}$. 
In order to find conditions among $b_r^{(\ell)}$ for $M_{i,i+1}^{(\ell)}$ to satisfy the positivity condition, one needs to know all the elements of the fused Temperley-Lieb generators $e_i^{(\ell;r)}$. We give the explicit forms of matrix elements of $e_i^{(\ell;r)}$ in Appendix~4. The expressions (\ref{elements-1})-(\ref{elements-3}) show that the operator $M_{i,i+1}^{(\ell;1)}$ have only positive off-diagonal elements, which ensures the existence of the multi-state ASEP for arbitrary $\ell$. 
For arbitrary combination of $e_i^{(\ell;r)}$, we have the following conjecture: 
\begin{conjecture} \label{conj:posi}
 Let us assume the update matrix takes the form given by (\ref{comb-l}). 
We normalize the parameters $b_r^{(\ell)}$ as $b_r^{(\ell)} = \begin{bmatrix} \ell \\ r \end{bmatrix} \widetilde{b}_r^{(\ell)}$. $\begin{bmatrix} x \\ y \end{bmatrix}$ is a $q$-binomial defined by $\begin{bmatrix} x \\ y \end{bmatrix} = \frac{[x]!}{[x-y]! [y]!}$, where $[x]! = [x][x-1] \cdots [1]$. Here we introduced a $q$-integer $[x] = \frac{q^x-q^{-x}}{q-q^{-1}}$. 
Since the update matrix is symmetric under $q \leftrightarrow q^{-1}$, we restrict our attention to the case $0 \leq q \leq 1$. Then, the restrictions among the parameters $\widetilde{b}_r^{(\ell)}, \dots, \widetilde{b}_{r+x}^{(\ell)}$ are obtained from the positivity condition of the matrix element $\langle 0| \otimes \langle r+x|\; M_{i,i+1}^{(\ell)} \;|r \rangle \otimes |x \rangle$ ($r=1,\dots,\ell-x$ and $x=0,\dots,\ell-1$), which leads to 
\begin{align}
 &(-1)^{r-1} \widetilde{b}_r^{(\ell)} > 0, \label{posi1}\\
 &\sum_{k=0}^s q^{(s-1)k}\; 
 \begin{bmatrix}
  s \\ k
 \end{bmatrix}
 \; \widetilde{b}_{k+i}^{(\ell)} (-1)^{i-1} > 0, \label{posi2}
% &(-1)^{r-1} b_r^{(\ell)} > 0, \\
% &\sum_{k=0}^s q^{(s-1)k} {_\ell}C_{k+i} \cdot {_{k+i}}C_1 \cdot b_{k+i}^{(\ell)} (-1)^{i-1} > 0,
  \quad
  i=1,\dots,\ell-s, 
\end{align}
for $s = 1,\dots,\ell-1$. 
\end{conjecture}

The explicit forms of $e_i^{(\ell;r)}$ are required also in the construction of matrix product states, which allow us to deal with open boundaries admitting particle incoming and outgoing. 
Recently, $T$-systems and $Y$-systems were derived for the fused Temperley-Lieb algebra~\cite{bib:MPR14}. We expect that these algebraic relations make it possible to obtain relations among the fused generators, as is the BMW algebra for the $\ell=3$ case. 
%%%%%%%%%%%%%%%%%%%%%%%%%%%%%%%%%%%%%%%%%%%%
\subsection{Special limits of multi-state ASEPs}
Now we consider two special limits of multi-state ASEPs; $q \to 1$ and $q \to 0$. The first one is known as the symmetric simple exclusion process (SSEP) for the two-state case, which has the same hopping rates to the right and to the left. On the other hand, the second limit is called the totally asymmetric simple exclusion process (TASEP), in which each particle is allowed to move only to the left. The special care is required to take the $q\to 0$ limit for avoiding singularities. 

%%%%
\subsubsection{$q \to 1$ limit}
First, we consider the $q \rightarrow 1$ limit of the multi-state ASEP. In this limit, the original ASEP becomes SSEP in which each particle hops to the right and to the left with the same rate as long as neighboring sites are unoccupied. 
In the multi-state case, for instance in the three-state case, we have the following update matrix: 
\begin{equation} \label{qto1}
 \underset{q \rightarrow 1}{\lim}\;M_{i,i+1}^{(2)}
  =
 %\left(
  \begin{pmatrix}
   0 & 0 & 0 & 0 & 0 & 0 & 0 & 0 & 0 \\
   0 & -\frac{1}{2} & 0 & \frac{1}{2} & 0 & 0 & 0 & 0 & 0 \\
   0 & 0 & -\beta -1 & 0 & \frac{1 + 2\beta}{4} & 0 & -\beta  & 0 & 0 \\
   0 & \frac{1}{2} & 0 & -\frac{1}{2} & 0 & 0 & 0 & 0 & 0 \\
   0 & 0 & 2 \beta +1 & 0 & -\beta -\frac{1}{2} & 0 & 2 \beta +1 & 0 & 0 \\
   0 & 0 & 0 & 0 & 0 & -\frac{1}{2} & 0 & \frac{1}{2} & 0 \\
   0 & 0 & -\beta  & 0 & \frac{1 + 2\beta}{4} & 0 & -\beta -1 & 0 & 0 \\
   0 & 0 & 0 & 0 & 0 & \frac{1}{2} & 0 & -\frac{1}{2} & 0 \\
   0 & 0 & 0 & 0 & 0 & 0 & 0 & 0 & 0 
  \end{pmatrix}_{i,i+1},
 %\right). 
\end{equation}
from which one easily obtains the symmetric hopping rates as in the two-state ASEP. 
In this limit, the similarity transformation matrix (\ref{U-mat}) becomes an identity matrix. On the other hand, the positivity condition (\ref{cond_comb}) is still valid; in $q \rightarrow 1$, the condition (\ref{cond_comb}) becomes $-\frac{1}{2} < \beta < 0$. Since we always have $\frac{1 + 2\beta}{4} < 2\beta + 1$ for $-\frac{1}{2} < \beta < 0$, the process of this limit forces to reduce the number of boxes occupied by two particles. However, for $-\frac{1}{2} < \beta < -\frac{1}{3}$, we have $2\beta + 1 < -\beta$, which implies that two particles in the same box tend to move together. 

Now let us discuss the arbitrary $(\ell+1)$-state case. Taking into account the fact that $q$-integers and $q$-binomials become ordinary integers and binomials in the $q \to 1$ limit, the matrix elements of fused TL generators (\ref{TL-ldim}) are computed as 
\begin{equation} \label{ssep-elements}
\begin{split}
 \langle r-x| \otimes \langle s+x|\; e_i^{(\ell;k)} \;|r \rangle \otimes |s \rangle
 &=
 (-1)^{x}
 \begin{pmatrix} \ell \\ r \end{pmatrix}^{-1}
 \begin{pmatrix} \ell \\ s \end{pmatrix}^{-1} \\
 &\cdot \sum_{j=x}^k 
 \begin{pmatrix} \ell-k \\ r-k+j \end{pmatrix}
 \begin{pmatrix} k \\ j-x \end{pmatrix}
 \begin{pmatrix} k \\ k-j \end{pmatrix}
 \begin{pmatrix} \ell-k \\ s-j  \end{pmatrix},  
% &\langle r_L| \otimes \langle s_L|\; e_i^{(\ell;r)}\; |r_R \rangle \otimes |s_R \rangle
% \\
% &= 
% \begin{pmatrix}
% \ell \\ r_L
% \end{pmatrix}^{-1}
% \begin{pmatrix}
%  \ell \\ s_L
% \end{pmatrix}^{-1}
% \delta_{r_L+s_L, r_R+s_R}
% \\
% &\cdot
% \sum_{\ell-r+1 \leq a_1 < \dots < a_r \leq \ell+r \atop \ell-r+1 \leq b_1 < \dots < b_r \leq \ell+r}
% (-1)^{{\rm card}(m_j \neq n_k)}
% \prod_{k=1}^r (1 - \delta_{a_{r-k+1}-\ell+r, \ell+r-a_k+1}) (1 - \delta_{b_{r-k+1}-\ell+r, \ell+r-b_k+1})
% \\
% &\cdot
% \sum_{1 \leq \bar{a}_1 < \dots < \bar{a}_{r_L - {\rm card}(a_k \leq \ell)} \leq \ell-r
% \atop{\ell+r+1 \leq \bar{a}_{r_L - {\rm card}(a_k \leq \ell)+1 < \dots < \bar{a}_{r_L+s_L-r} \leq 2\ell}}}
% %
% \sum_{1 \leq \bar{b}_1 < \dots < \bar{b}_{r_R - {\rm card}(b_k \leq \ell)} \leq \ell-r
% \atop{\ell+r+1 \leq \bar{b}_{r_R - {\rm card}(b_k \leq \ell)+1 < \dots < \bar{b}_{r_L+s_L-r} \leq 2\ell}}}
% %
% \prod_{k=1}^{r_L+s_L-r} \delta_{\bar{a}_k\bar{b}_k}.
\end{split}
\end{equation}
where $\begin{pmatrix} x \\ y \end{pmatrix}$ is an ordinary binomial. 
The positivity conditions are obtained from Conjecture \ref{conj:posi} by taking the $q \to 1$ limit as 
\begin{align}
 &(-1)^{r-1} \widetilde{b}_r^{(\ell)} > 0, \\
 &\sum_{k=0}^s 
 \begin{pmatrix}
  s \\ k
 \end{pmatrix}
 \; \widetilde{b}_{k+i}^{(\ell)} (-1)^{i-1} > 0, 
% &(-1)^{r-1} b_r^{(\ell)} > 0, \\
% &\sum_{k=0}^s q^{(s-1)k} {_\ell}C_{k+i} \cdot {_{k+i}}C_1 \cdot b_{k+i}^{(\ell)} (-1)^{i-1} > 0,
  \quad
  i=1,\dots,\ell-s, 
\end{align}
for $s = 1,\dots,\ell-1$.

%%%%%%%%%%%%%%%%%%%%%%%%%%%%%%%%%%%%%%%%%%%%%
\subsubsection{$q \rightarrow 0$ limit}
Next, we consider the $q \to 0$ limit of the multi-state ASEP. In order to take this limit, one needs to take care of singularities. By setting $q = x/y$ and then taking $x \to 0$ after removing singular parts, we obtain the well-defined $q \to 0$ limit of $M_{i,i+1}^{(\ell)}$. By rescaling the parameter $\beta$ as $\beta = \frac{q^{2}}{q + q^{-1}} \widetilde{\beta}$, we obtain the following update matrix: 
\begin{equation}
 \underset{q \rightarrow 0}{\lim}\; q^3 (q+q^{-1})^2 M_{i,i+1}^{(2)}
  =
%  \left(
  \begin{pmatrix}
   0 & 0 & 0 & 0 & 0 & 0 & 0 & 0 & 0 \\
   0 & -1 & 0 & 0 & 0 & 0 & 0 & 0 & 0 \\
   0 & 0 & -1 & 0 & 0 & 0 & 0 & 0 & 0 \\
   0 & 1 & 0 & 0 & 0 & 0 & 0 & 0 & 0 \\
   0 & 0 & \widetilde{\beta} +1 & 0 & -1 & 0 & 0 & 0 & 0 \\
   0 & 0 & 0 & 0 & 0 & -1 & 0 & 0 & 0 \\
   0 & 0 & -\widetilde{\beta}  & 0 & 1 & 0 & 0 & 0 & 0 \\
   0 & 0 & 0 & 0 & 0 & 1 & 0 & 0 & 0 \\
   0 & 0 & 0 & 0 & 0 & 0 & 0 & 0 & 0 
  \end{pmatrix}_{i,i+1}, 
%  \right), 
\end{equation}
which describes the TASEP-like transition, namely, particles are allowed to move only to the left. The positivity condition is satisfied if $\widetilde{\beta}$ obeys $-1 < \widetilde{\beta} < 0$. This parameter determines the strength of the coupling of two particles in the same box; for larger $|\widetilde{\beta}|$, more hoppings of two particles at the same time, while for smaller $|\widetilde{\beta}|$, we obtain more hoppings of each particle separately. 

The result indicates that the parameter $\widetilde{\beta}$ acts as friction between particles. That is, small $\widetilde{\beta}$ describes dynamics of particles with small friction, while large $\widetilde{\beta}$ describes particle dynamics with large friction. Thus, one application of the multi-state TASEP would be a granular model with controllable friction parameters. 

For the well-definedness of the $(\ell+1)$-state TASEP, we give the following conjecture: 
\begin{conjecture}
 The following reparametrization of $b_r^{(\ell)}$ in (\ref{comb-l}) gives a proper TASEP limit of the multi-state ASEP: 
\begin{equation}
 M_{i,i+1}^{(\ell)} = \sum_{r=1}^{\ell} q^{r \ell} 
  \begin{bmatrix} \ell \\ r \end{bmatrix}
  \widetilde{b}_r^{(\ell)} M_{i,i+1}^{(\ell;r)}. 
\end{equation}
This choice of parametrization leads to the multi-state TASEP process with hopping rates $w(i,j|i-x,j+x) = (-1)^{x} \widetilde{b}_{x}^{(\ell)}$, where $x=1,\dots,\ell$. The positivity condition for $\widetilde{b}_r^{(\ell)}$ is given by $0 < \widetilde{b}_{x}^{(\ell)} < 1$ for even $x$ and $-1 < \widetilde{b}_x^{(\ell)} < 0$ for odd $x$. 
\end{conjecture}

 Thus, the parameters $\widetilde{b}_r^{(\ell)}$ are independently controllable within $(0,1)$ for even $r$ and $(-1,0)$ for odd $r$. The independent controllability is a property not only for $q \to 0$ but also for arbitrary $q$, as long as $\widetilde{b}_r^{(\ell)}$ satisfy the conditions (\ref{posi1}) and (\ref{posi2}).

%%%%%%%%%%%%%%%%%%%%%%%%%%%%%%%%%%%%%%%%%%%%%
%%%%%%%%%%%%%%%%%%%%%%%%%%%%%%%%%%%%%%%%%%%%
\section{Steady states}
In this section, we discuss the steady state of multi-state ASEPs constructed in the previous section. Due to invariance of the TL algebra with respect to a transformation $q \leftrightarrow q^{-1}$, we discuss only the case of $q \geq 1$. This condition is physically interpreted as more hopping to the right than to the left.

%Since a steady state $|P_{\rm st}(t) \rangle$ is invariant under time development, it satisfies 
%\begin{equation} \label{steady}
% \frac{d}{dt} |P_{\rm st}(t) \rangle = 0. 
%\end{equation}

A steady state and physical quantities of the closed two-state ASEP have been intensively studied in~\cite{bib:SS94} based on the $U_q(sl_2)$ algebraic relations. Starting from the zero-particle state, which is obviously a steady state of the ASEP, they subsequently obtained an $n$-particle state for arbitrary $n$ by applying the $S^-$-operator $n$ times to the zero-particle state. 
A series of steady sates of the multi-state ASEP is derived in a similar way from the zero-particle state, as the update matrix still commutes with the $U_q(sl_2)$ generators. 
However, multi-state extension requires much more cumbersome computation than the two-state case even for the norms. In order to resolve this difficulty, we use the property of the projection operator $Y^{(\ell)}|0\rangle = |0\rangle$ and a map from $(\ell+1)$-dimensional representations onto the fundamental representations of the $U_q(sl_2)$ algebra, which allows us to proceed all calculations of the $(\ell+1)$-state ASEP in terms of the two-state system with $\ell$ times the length. 

%Since we constructed the update matrix of the multi-state ASEPs in order to be invariant under the $U_q(sl_2)$ algebra, the steady state is given by the zero eigenvectors of the $U_q(sl_2)$ generators up to the similarity transformations. 
%
\begin{proposition}
Let us consider the $(\ell+1)$-state ASEP on $N$ sites. The vacuum (zero-particle) state is a steady state of the $(\ell+1)$-state ASEP and mapped onto the vacuum of the two-state ASEP with length $\ell N$: 
\begin{equation} \label{zero_vector}
\begin{split}
 &|\ell; 0 \rangle_{1,\dots,N} \mapsto \otimes^{\ell N} \begin{pmatrix} 1 \\ 0 \end{pmatrix} = |0 \rangle_{1,\dots,\ell N}, \\
 &{_{1,\dots,N}\langle} \ell; 0| \mapsto \otimes^{\ell N} \begin{pmatrix} 1 & 0 \end{pmatrix} = {_{1,\dots\ell N}\langle} 0|. 
\end{split}
\end{equation}
\end{proposition}
\begin{proof}
It is obvious that the zero-particle state is invariant under time development. 
Since we defined the empty state and the occupied state of the two-state ASEP as in (\ref{TL-ASEP}), the zero-particle state of the two-state ASEP with length $N$ is given by an $N$-fold tensor product of the two-dimensional highest weight vectors: 
\begin{equation}
 |0 \rangle_{1,\dots,N} = \otimes^N \begin{pmatrix} 1 \\ 0 \end{pmatrix}. 
\end{equation}
Similarly, the vacuum of the $(\ell+1)$-state ASEP is given by an $N$-fold tensor product of the $(\ell+1)$-dimensional highest weight vectors. Taking into account that the $(\ell+1)$-dimensional highest weight vector is expressed by an $\ell$-fold tensor product of the two-dimensional highest weight vectors with the projection operator, the following relation holds: 
\begin{equation}
 |\ell; 0 \rangle_{1,\dots,N} \mapsto \otimes^N \left(Y^{(\ell)} \otimes^{\ell} \begin{pmatrix} 1 \\ 0 \end{pmatrix}\right). 
\end{equation}
Since the projection operator $Y^{(\ell)}$ trivially acts on the highest-weight vector, we obtain the relation (\ref{zero_vector}). 

The dual vector of the zero-particle state is obtained in the same way. 
\end{proof}

From now on, we abbreviate the subscripts $1,\dots,N$ attached to each vector unless confusion occurs. An $n$-particle steady state of the two-state ASEP is then subsequently created from the vacuum state by applying the operator $U^{(\ell)} \Delta^{(N)}(S^{\pm}) (U^{(\ell)})^{-1}$~\cite{bib:SS94}. Since the similarity transformation operator trivially acts on the vacuum $(U^{(\ell)})^{-1} |\ell;0\rangle = |\ell;0\rangle$, we define the following ``untwisted'' $n$-particle state: 
\begin{equation} \label{np-vector}
 |n \rangle = \frac{1}{[n]!} (\Delta^{(N)} (S^-))^n |0 \rangle, 
  \qquad
  \langle n| = \frac{1}{[n]!} \langle 0| (\Delta^{(N)}(S^+))^n. 
\end{equation}
Here we used the $q$-factorials $[n]! = [n] [n-1] \dots [1]$ consisting of the $q$-integers defined by 
\begin{equation}
 [n] = \frac{q^n - q^{-n}}{q - q^{-1}}. 
\end{equation}
\begin{proposition} \label{prop:basis}
 An $n$-particle steady state of the $(\ell+1)$-state ASEP is mapped onto a tensor product of two-dimensional vectors: 
\begin{equation} \label{map:vector}
 |\ell; n \rangle \mapsto \frac{1}{[n]!} (\Delta^{(\ell N)} (S^-))^{n} |0 \rangle,
  \qquad
  \langle \ell; n| \mapsto \frac{1}{[n]!} \langle 0| (\Delta^{(\ell N)} (S^+))^{n}. 
\end{equation}
\end{proposition}
\begin{proof}
An untwisted $n$-particle state of the $(\ell+1)$-state ASEP is generated by applying the operator $S^{\pm}$ to the vacuum: 
\begin{equation} \label{lstate-nparticle}
 |\ell;n \rangle = \frac{1}{[n]!} (\Delta^{(N)} (S^{-,(\ell)}))^n |\ell;0 \rangle, 
  \qquad
  \langle \ell;n| = \frac{1}{[n]!} \langle \ell;0| (\Delta^{(N)} (S^{+,(\ell)}))^n. 
\end{equation}
Here $S^{\pm, (\ell)}$ are the $U_q(sl_2)$ generator of $(\ell+1)$-dimensional representations. 

Let us remind that $U_q(sl_2)$ generators have spatial extension called co-multiplication: 
\begin{equation}
 \Delta^{(\ell)}(S^{\pm}) = \sum_{j=1}^{\ell} q^{S_1^z} \otimes \cdots \dots \otimes q^{S_{j-1}^z} \otimes S_j^{\pm} \otimes q^{-S_{j+1}^z} \otimes \dots \otimes q^{-S_{\ell}^z}. 
\end{equation}
From the definition, $\Delta^{(\ell)}(S^{\pm})$ $q$-symmetrically acts on the $\ell$-fold tensor product spaces. 
On the other hand, an $(\ell+1)$-dimensional vector space associated with the $U_q(sl_2)$ algebra is mapped onto a $q$-symmetric $\ell$-fold tensor product of two-dimensional vector spaces. Therefore, the operators $S^{\pm, (\ell)}$ admit the following map: 
\begin{equation} \label{lstate-operator}
S^{\pm, (\ell)} \mapsto \Delta^{(\ell)}(S^{\pm}). 
\end{equation}
Subsequently, coproduct of $U_q(sl_2)$ generators are invariant under the $q$-symmetrizer, {\it i.e.} the projection operator: 
\begin{equation}
 Y^{(\ell)} \Delta^{(\ell)}(S^{\pm}) Y^{(\ell)} = \Delta^{(\ell)}(S^{\pm}). 
\end{equation}

Combining (\ref{lstate-nparticle}) and (\ref{lstate-operator}), we obtain
\begin{equation}
\begin{split}
 |\ell;n \rangle 
 &\mapsto 
 \frac{1}{[n]!}\Big( (\underbrace{ \Delta^{(N)} \circ \dots \circ \Delta^{(N)}}_{\ell}) \, 
 \Delta^{(\ell)}(S^-)
 \Big)^n |0\rangle 
\\
 &= \frac{1}{[n]!} (\Delta^{(\ell N)}(S^-))^n |0\rangle, 
\end{split}
\end{equation}
which reads (\ref{map:vector}). Here we have chosen the normalization constant as $1/[n]!$. 
\end{proof}
%
%%%%%%%%%%%%%%%%%%%%%%%%%%%%%%%%%%%%%%%%%%%%%%%%%%%%%%%%%%%%%%%%%%%%%%%%%%%%%%%
%%%%%%%%%%%%%%%%%%%%%%%%%%%%%%%%%%%%%%%%%%%%%%%%%%%%%%%%%%%%%%%%%%%%%%%%%%%%%%%
\subsection{Norms of the steady states}
In derivation of steady states, we showed that an $n$-particle state of the $(\ell+1)$-state ASEP is expressed in terms of two-dimensional $U_q(sl_2)$ basis. 
Although it is often difficult to proceed analytical computation on higher-dimensional representations, now we can use various formulae having already achieved for the two-state case. 

For instance, the norm of an $n$-particle steady state of the two-state ASEP was computed in~\cite{bib:SS94}: 
\begin{equation}
 \langle n|U U^{-1}|n \rangle = \langle n|n \rangle 
    = \begin{bmatrix} N \\ n \end{bmatrix} = \frac{[N]!}{[N-n]! [n]!}. 
\end{equation}
Using relations (\ref{map:vector}), we straightforwardly obtain the multi-state extension of the norm. 
\begin{proposition}
The norm of an $n$-particle steady state of the $(\ell+1)$-state ASEP is given by  
\begin{equation} \label{norm}
\begin{split}
 \langle \ell; n|U^{(\ell)} (U^{(\ell)})^{-1}|n; \ell \rangle 
 &= \langle \ell; n|\ell; n \rangle \\
 &= \frac{1}{[n]![n]!} \langle 0| (\Delta^{(\ell N)}(S^+))^n (\Delta^{(\ell N)}(S^-))^n |0 \rangle \\
 &= \begin{bmatrix} \ell N \\ n \end{bmatrix}. 
% &= \prod_{k=1}^n \left\{\sum_{j = 1}^k \frac{q^{2(N-k+j)} - q^{-2(N-k+j)}}{q - q^{-1}}\right\}. 
\end{split}
\end{equation}
\end{proposition}
\begin{proof}
The first equation is proved from Proposition \ref{prop:basis}. The second equation is derived from the commutation relations among $U_q(sl_2)$ generators: 
\begin{equation}
 [\Delta(S^+),\, \Delta(S^-)] = \frac{\Delta(q^{2S^z}) - \Delta(q^{-2S^z})}{q - q^{-1}}. 
\end{equation}
Detail proof is given in Appendix~5.  
\end{proof}
Here we define the normalized vector by 
\begin{equation} \label{normal}
 |\ell; n\rangle_{\rm norm} \equiv \begin{bmatrix} \ell N \\ n \end{bmatrix}^{-1} |\ell; n\rangle. 
\end{equation}
in order to have a normalization condition: $\langle \ell; n|\ell; n \rangle_{\rm norm} = 1$.

%%%%%%%%%%%%%%%%%%%%%%%%%%%%%%%%%%%%%%%%%%%%%%%%%%%%%%%%%%%%%%%%%%%%%%%%%%%%%%%%%%%%
\subsection{Particle-density profiles}
In this subsection, we derive particle-density profiles of an $n$-particle steady state of the multi-state ASEP. Particle-density profiles of the two-state ASEP have been closely studied and known to show transitions under general boundary conditions~\cite{bib:K91, bib:DDM92, bib:DE93, bib:SD93}. We show, although our boundary conditions are closed ones on which the number of particles is conserved, how particle-density profiles depend on the number of states of the process. 

Since we have expressed an arbitrary steady state of the $(\ell+1)$-state ASEP by that of the two-state ASEP, correlation functions of the multi-state ASEP are also computed on the fundamental representations if physical quantities can be also mapped to the two-dimensional representations. 

Therefore, the first aim of this subsection is to write an $(\ell+1)$-by-$(\ell+1)$ matrix by a tensor product of two-by-two matrices. By means of the method used in~\cite{bib:DM08}, the following proposition holds for an $(\ell+1)$-by-$(\ell+1)$ single-entry matrix with the $(r,s)$-entry $1$. 
\begin{proposition} \label{prop:reduction}
An $(\ell+1)$-by-$(\ell+1)$ single-entry matrix with $1$ at the $(r,s)$-entry is written by a vector basis and its dual $|\ell;r-1\rangle \otimes \langle \ell;s-1|$. Using Proposition \ref{prop:basis}, vector basis are decomposed into an $\ell$-fold tensor product of the two-dimensional vector basis, and thus a single-entry matrix of an $(\ell+1)$-dimensional representation is expressed by an $\ell$-fold tensor product of the fundamental representations. 
\end{proposition}
Using this proposition, we compute particle-density profiles of the multi-state ASEP. For simplicity, we first discuss the three-state case and then show a general case later.

\subsubsection{Density profiles on the three-state ASEP}
Under the presence of $n$ particles, particle density at the $x$th site is given by the following quantity: 
\begin{equation} \label{def:dense}
\begin{split}
 \rho^{(2)}_n(x) &= \langle 2; n| U^{(2)}
  {\rm diag.} (0,1,2)_x
(U^{(2)})^{-1} |2; n \rangle_{\rm norm} \\
 &= \langle 2; n| 
  {\rm diag.} (0,1,2)_x
 |2; n \rangle_{\rm norm}, 
\end{split}
\end{equation}
where ${\rm diag.}(0,1,2)_x$ is a three-by-three diagonal matrix which acts nontrivially on the $x$th space of an $N$-fold tensor product of three-dimensional vector spaces. 

Using Proposition \ref{prop:reduction}, the matrix ${\rm diag.}(0,1,2)$, which counts the number of particles at the $x$th site in a steady state, is written by 
\begin{equation} \label{tensor:dense}
\begin{split}
&{\rm diag.} (0,1,2)
 = 1 \cdot |2; 1 \rangle_{\rm norm} \otimes \langle 2; 1|
	+ 2 \cdot |2; 2 \rangle_{\rm norm} \otimes \langle 2; 2|. 
%& \mapsto \frac{1}{[2]} \Big\{
%		q \begin{pmatrix} 1 & \\ & 0 \end{pmatrix} \otimes \begin{pmatrix} 0 & \\ & 1 \end{pmatrix}
%		+ \begin{pmatrix} & 1 \\ & \end{pmatrix} \otimes \begin{pmatrix} & \\ 1 & \end{pmatrix} \\
%		&\hspace{12mm}
%		+ \begin{pmatrix} & \\ 1 & \end{pmatrix} \otimes \begin{pmatrix} & 1 \\ & \end{pmatrix}
%		+ q^{-1} \begin{pmatrix} 0 & \\ & 1 \end{pmatrix} \otimes \begin{pmatrix} 1 & \\ & 0 \end{pmatrix}
%	\Big\} + 2 \begin{pmatrix} 0 & \\ & 1 \end{pmatrix} \otimes \begin{pmatrix} 0 & \\ & 1 \end{pmatrix}.
\end{split}
\end{equation}
Using the relations (\ref{map:vector}), the right-hand side of (\ref{tensor:dense}) is expressed in terms of the two-dimensional $U_q(sl_2)$ vector basis. For instance, $|2;2\rangle_{\rm norm} \otimes \langle 2;2|$ is expressed as 
\begin{equation}
\begin{split}
|2;2\rangle_{x; {\rm norm}} \otimes {_x \langle} 2;2| 
 &\mapsto |1 \rangle_{2x-1; {\rm norm}} |1 \rangle_{2x; {\rm norm}} \otimes {_{2x-1}}\langle 1| {_{2x}}\langle 1| \\
 &= {\rm diag.}(0,1)_{2x-1} \otimes {\rm diag.}(0,1)_{2x} \\
 &= n_{2x-1} n_{2x}. 
\end{split}
\end{equation}
Hence we have 
\begin{equation} \label{2states:density}
\begin{split}
\rho^{(2)}_n(x) =& \langle n| \frac{1}{[2]} \Big(
	q (1 - n_{2x-1}) n_{2x} + S_{2x-1}^+ S_{2x}^- + S_{2x-1}^- S_{2x}^+ + q^{-1} n_{2x-1} (1 - n_{2x})
\Big)
|n \rangle_{\rm norm} \\
	&+ 2\cdot \langle n|
	n_{2x-1} n_{2x}
|n \rangle_{\rm norm}. 
\end{split}
\end{equation}
The first line of (\ref{2states:density}) is written only by the particle-counting operators using the relations (\ref{creation}): 
\begin{equation}
\begin{split}
&\frac{q}{[2]} \cdot \langle n| n_{2x} |n \rangle_{\rm norm} + \frac{q^{-1}}{[2]} \cdot \langle n| n_{2x-1} |n \rangle_{\rm norm}
	- \langle n| n_{2x-1} n_{2x} |n \rangle_{\rm norm} \\
	&+ \frac{2}{[2]} \frac{[n] q^{-N+2\cdot (2x - 1)}}{[N-n+1]} \cdot 
		\langle n-1| ( 1 - n_{2x-1} )( 1 - n_{2x} ) |n-1 \rangle_{\rm norm}. 
\end{split}
\end{equation}
The second line of (\ref{2states:density}) is decomposed into a summation of one-point functions using the formula (\ref{n-point}). Then, we obtain that particle-density profiles of the three-state ASEP are expressed in terms of one-point functions: 
\begin{equation} \label{2p-corr}
\begin{split}
\rho^{(2)}_n(x) = \langle n| n_{2x-1} |n \rangle_{\rm norm} 
	+ \langle n| n_{2x} |n \rangle_{\rm norm}, 
\end{split}
\end{equation}
whose explicit expression is evaluated from the formula (\ref{1p-corr}): 
\begin{equation}
\rho^{(2)}_n(x) = \begin{bmatrix} 2N \\ n \end{bmatrix}^{-1}
	\sum_{k=0}^{n-1} (-1)^{n-k+1} \begin{bmatrix} 2N \\ k \end{bmatrix} q^{-(n-k) (2N-4x+2)} 
	\Big(
		q^{n-k} + q^{-(n-k)}
	\Big). 
\end{equation}
The resulting relation (\ref{2p-corr}) is understood as follows; We have constructed three-dimensional basis of the $U_q(sl_2)$ algebra by $q$-symmetrizing a two-fold tensor product of two-dimensional $U_q(sl_2)$ vector spaces. It is equivalent to work on the double length of the two-state ASEP but with $q$-symmetrizing $(2x-1)$th and $2x$th sites, instead of working on three-state ASEP. As a result, particle density at the $x$th site of the three-state model is given by a summation of particle densities at the $(2x-1)$th site and the $2x$th site of the two-state ASEP. 
\subsubsection{Density profiles on the $(\ell+1)$-state ASEP}
%The key point to compute particle-density profiles of the ASEP with arbitrary number of states is the existence of the mapping of vector spaces of the $U_q(sl_2)$ algebra of arbitrary dimensional representations onto those of fundamental representations, on which the computations of physical quantities are drastically simplified. 
As the analogue of the three-state ASEP, particle-density profiles of the $(\ell+1)$-state ASEP is also expressed by a summation of one-point functions of the two-dimensional representations. Now we prove the following proposition: 
\begin{proposition}
Particle-density profiles of the $(\ell+1)$-state ASEP are evaluated through the following expression:  
\begin{equation} \label{density}
\rho^{(\ell)}_n(x) = \sum_{j=1}^{\ell} \langle n|(U^{(\ell)})^{-1} n_{\ell (x-1)+j} U^{(\ell)} |n \rangle_{\rm norm}
 = \sum_{j=1}^{\ell} \langle n| n_{\ell (x-1)+j}|n \rangle_{\rm norm}. 
\end{equation} 
\end{proposition}
In the proof of this proposition, we need the following proposition besides Proposition \ref{prop:n-point}: 
\begin{proposition}
An expectation value of a pair of the spin operators $S^{\pm}$ on the $n$-particle steady state is written by particle-counting operators in two different ways: 
\begin{align}
\langle n| S_{x_1}^{\pm} S_{x_2}^{\mp} |n \rangle_{\rm norm}
&= q^{x_2 - x_1} \cdot \langle n| n_{x_1} (1- n_{x_2}) |n \rangle_{\rm norm} \label{Spm-1}\\
&= q^{-(x_2 - x_1)} \cdot \langle n| (1 - n_{x_1}) n_{x_2} |n \rangle_{\rm norm}. \label{Spm-2}
\end{align}
\end{proposition}
\begin{proof}
Since the spin operators act on an $n$-particle steady state as (\ref{creation}), the following relation holds between an expectation value of the spin operators and that of the particle-counting operators: 
\begin{equation}
\begin{split}
\langle n| S_{x_1}^{\pm} S_{x_2}^{\mp} |n \rangle_{\rm norm} 
&= \begin{bmatrix} N \\ n \end{bmatrix}^{-1} \cdot \langle n| S_{x_1}^{\pm} S_{x_2}^{\mp} |n \rangle \\
&= \begin{bmatrix} N \\ n \end{bmatrix}^{-1} q^{(-N - 1 + x_1 + x_2)} \cdot \langle n-1| (1-n_{x_1}) (1-n_{x_2}) |n \rangle \\
&= \frac{[n]}{[N - n + 1]} q^{-N - 1 + x_1 + x_2} \cdot \langle n-1| (1-n_{x_1}) (1-n_{x_2}) |n \rangle_{\rm norm}. 
\end{split}
\end{equation}
Applying the recursion relation (\ref{l-corr}) to the operator $(1-n_{x_1})$, one obtains the relations (\ref{Spm-1}), while by applying (\ref{l-corr}) to $(1-n_{x_2})$, the relation (\ref{Spm-2}) is obtained. 
\end{proof}

Let us remind that Proposition \ref{prop:reduction} allows us to map the $(j,j)$-element of a $(\ell+1)$-by-$(\ell+1)$ matrix onto an $\ell$-fold tensor product of two-by-two matrices, which is written by a tensor product of two-dimensional vector spaces: 
\begin{equation}
 |\ell; j \rangle_{\rm norm} \otimes \langle \ell; j|
 = |j \rangle_{1,\dots,\ell; \rm norm} \otimes {_{1,\dots,\ell} \langle} j|. 
\end{equation}
Then we obtain the following lemmas. 
\begin{lemma}
Probability to obtain $j$ particles at the $x$th site in an $n$-particle steady state is calculated as 
\begin{equation}
\begin{split}
	&\langle n| 
 \Big( |\ell; j \rangle_{x;\, {\rm norm}} \otimes {_x \langle} \ell; j| \Big) 
 |n \rangle_{\rm norm} \\
	&= \sum_{\{x_1, \dots, x_j\} \in \mathfrak{S}_\ell \setminus \mathfrak{S}_j}
	   \langle n| \prod_{k=1 \atop k \neq x_1,\dots, x_j}^{\ell} (1-n_{\ell (x-1)+k}) \prod_{i=1}^j n_{\ell (x-1)+x_i} |n \rangle_{\rm norm}. 
\end{split}
\end{equation}
\end{lemma}
\begin{lemma}
Particle density at the $x$th site of the $(\ell+1)$-state ASEP is expressed by the fundamental representation as follows: 
\begin{equation}
\begin{split}
	\rho^{(\ell)}_n(x) &= \sum_{j=0}^{\ell} j \cdot \langle n|\ell;j \rangle_{x; {\rm norm}} \otimes {_x \langle} \ell;j|n \rangle \\
	&= \sum_{j=0}^{\ell} 
	\sum_{\{x_1, \dots, x_j\} \in \mathfrak{S}_\ell \setminus \mathfrak{S}_j}
	   \langle n| \prod_{k=1 \atop k \neq x_1,\dots, x_j}^{\ell} (1-n_{\ell (x-1)+k}) \prod_{i=1}^j n_{\ell (x-1)+x_i} |n \rangle_{\rm norm} \\
 &= \sum_{p=1}^{\ell} \left(
 \sum_{r=0}^{p-1} (-1)^{r} (p-r) 
 \begin{pmatrix}
  p \\ p-r
 \end{pmatrix}
 \sum_{\{x_1,\dots,x_p\} \in \mathfrak{S}_\ell \setminus \mathfrak{S}_p} {_{\ell N} \langle} n| \prod_{j=1}^{p} n_{\ell (x-1)+x_j} |n \rangle_{\ell N; {\rm norm}}
 \right). 
\end{split}
\end{equation}
\end{lemma}
Using the following relation: 
\begin{equation}
	\sum_{r=0}^{p-1} (-1)^{r} (p-r) 
	 \begin{pmatrix}
	  p \\ p-r
	 \end{pmatrix} = 0
	\qquad p \in \mathbb{Z}_{\ge 2}, 
\end{equation}
we obtain (\ref{density}). From the formula for a one-point function (\ref{1p-corr}) in Appendix~6, we finally obtain an expression for the particle-density profile of the $(\ell+1)$-state ASEP: 
\begin{equation} \label{density*}
 \rho^{(\ell)}_n (x) = 
  \sum_{j = 1}^{\ell} \begin{bmatrix} \ell N \\ n  \end{bmatrix}^{-1}
  \sum_{k = 0}^{n - 1} (-1)^{n - k + 1} \begin{bmatrix} \ell N \\ k \end{bmatrix}
  q^{-(n-k) (\ell N + 1 - 2(\ell (x-1) + j))}. 
\end{equation}

The Fig.~\ref{fig:density} shows the particle-density profiles of steady states of the two, three, and four-state ASEPs. The profiles show the step-function like behaviors with decay lengths inversely proportional to the number of states of the processes. The step-function like behaviors are phenomenologically understood as a result that we chose a bigger hopping rate to the right than to the left. 
Detailed analysis of decay lengths is given later with asymptotic analysis. 
\begin{figure}
\begin{center}
 \includegraphics[scale=0.75]{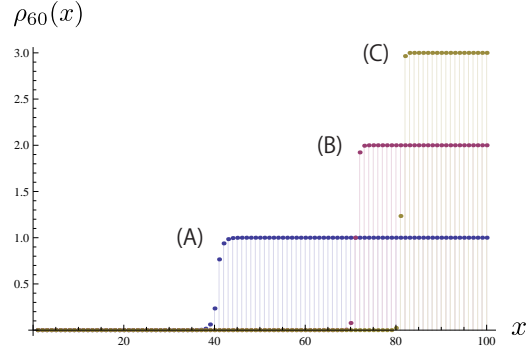}
\caption{Particle-density profiles of the (A) two-state ASEP, (B) three-state ASEP, and (C) four-state ASEPs. The plots are given for a system with size $N=100$ under the presence of $n=60$ particles with $q=2$. The profiles show the step-function like behaviors. If one chooses $q^{-1}=2$, one obtains plots reflected at $x=50$. } \label{fig:density}
\end{center}
\end{figure}
%
%%%%%%%%%%%%%%%%%%%%%%%%%%%%%%%%%%%%%%%%%%%%%%%%%%%%%%%%%%%%%%%%%%
%%%%%%%%%%%%%%%%%%%%%%%%%%%%%%%%%%%%%%%%%%%%%%%%%%%%%%%%%%%%%%%%%%
\subsection{Currents}
In this subsection, we calculate particle currents on the multi-state ASEPs. 
At the $x$th site of the two-state ASEP, current is defined through the following quantity: 
\begin{equation}
 \begin{split}
  &J(x) := J_R(x) - J_L(x), \\
  &J_R(x) := q \langle n|U^{-1} n_x (1 - n_{x+1}) U|n \rangle_{\rm norm} = q \langle n| n_x (1 - n_{x+1}) |n \rangle_{\rm norm}, \\
  &J_L(x) := q^{-1} \langle n|U^{-1} (1 - n_x) n_{x+1} U|n \rangle_{\rm norm} = q^{-1} \langle n| (1 - n_x) n_{x+1} |n \rangle_{\rm norm}. 
 \end{split}
\end{equation}
By definition, $J_R(x)$ gives an expectation value for the $x$th site being occupied at the same time with the $(x+1)$th site being empty, while $J_L(x)$ gives an expectation value for the $(x+1)$th site being occupied at the same time with the $x$th site being empty up to overall factors (Fig.~\ref{fig:currents}). 
\begin{figure}
\begin{center}
 \includegraphics[scale=0.7]{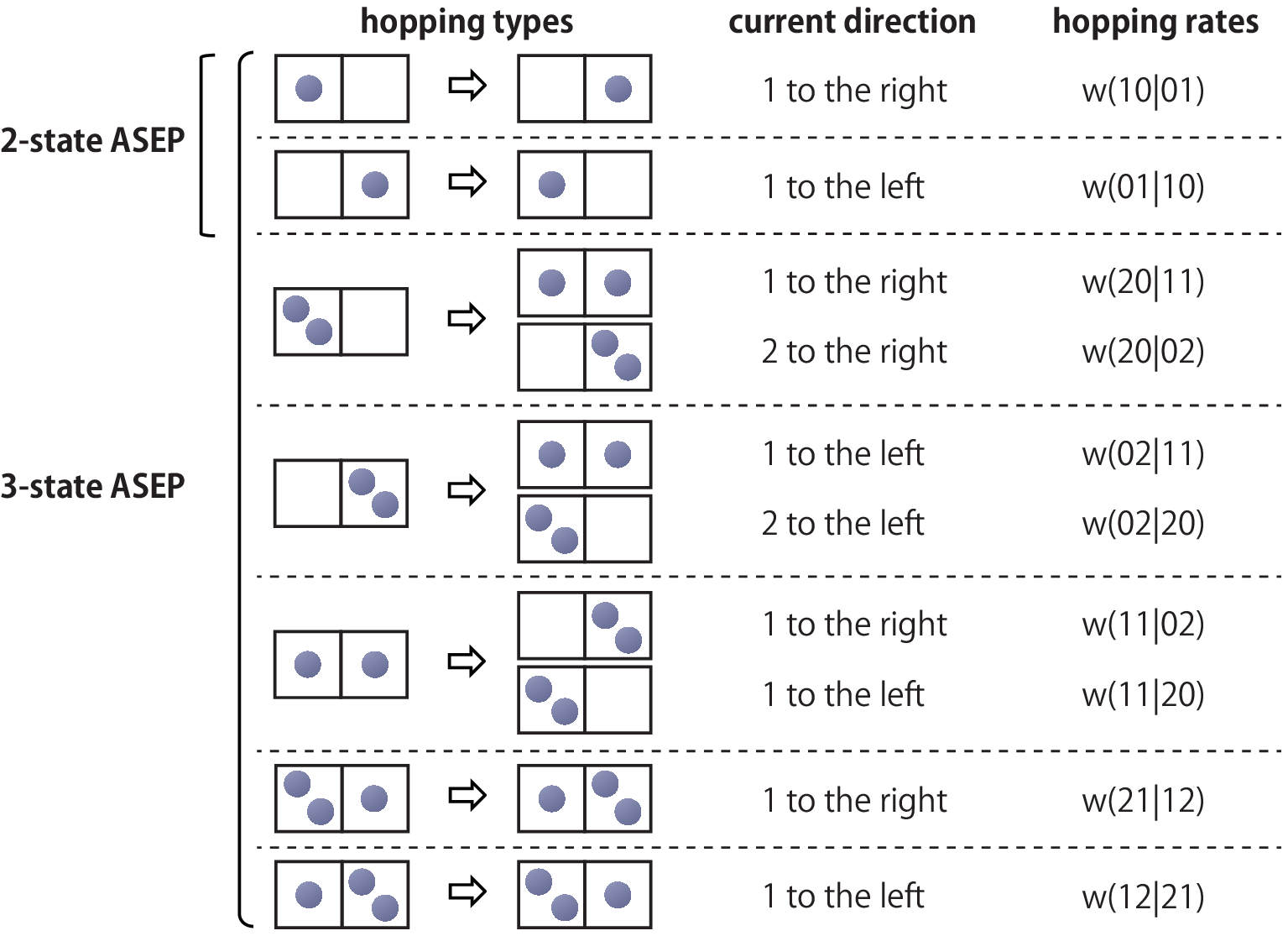}
 \caption{Allowed configurations of particles and hopping types. Five different types of right-hopping exist, while another five types for left-hopping in the three-state model. The number of particles which move together is also denoted. Explicit forms of hopping rates $w(ab | a'b')$ are given in (\ref{hopping_rates}).} \label{fig:currents}
\end{center}
\end{figure}
However, in the multi-state case, there are several possible ways for a particle to move to the right or left; For instance, in the three-state ASEP, there are five types of different hoppings which contribute to right-moving currents and another five to left-moving currents (Fig.~\ref{fig:currents}). Since our Markov matrix of the three-state ASEP is given by (\ref{3state-markov}), we define currents on the $n$-particle steady state of the three-state ASEP by 
\begin{equation} \label{currents4}
\begin{split}
 J_R^{(2)}(x) 
&= 
w(10 | 01) \langle 2;n| \left(\begin{smallmatrix} 0&& \\ &1& \\ &&0 \end{smallmatrix}\right)_{x} \left(\begin{smallmatrix} 1&& \\ &0& \\ &&0 \end{smallmatrix}\right)_{x+1} |2;n \rangle_{\rm norm}
+ w(20 | 11) \langle 2;n| \left(\begin{smallmatrix} 0&& \\ &0& \\ &&1 \end{smallmatrix}\right)_{x} \left(\begin{smallmatrix} 1&& \\ &0& \\ &&0 \end{smallmatrix}\right)_{x+1} |2;n \rangle_{\rm norm} \\
&+ w(11 | 02) \langle 2;n| \left(\begin{smallmatrix} 0&& \\ &1& \\ &&0 \end{smallmatrix}\right)_{x} \left(\begin{smallmatrix} 0&& \\ &1& \\ &&0 \end{smallmatrix}\right)_{x+1} |2;n \rangle_{\rm norm}
+ w(21 | 12) \langle 2;n| \left(\begin{smallmatrix} 0&& \\ &0& \\ &&1 \end{smallmatrix}\right)_{x} \left(\begin{smallmatrix} 0&& \\ &1& \\ &&0 \end{smallmatrix}\right)_{x+1} |2;n \rangle_{\rm norm}
\\
&+ 2 w(20 | 02) \langle 2;n| \left(\begin{smallmatrix} 0&& \\ &0& \\ &&1 \end{smallmatrix}\right)_{x} \left(\begin{smallmatrix} 1&& \\ &0& \\ &&0 \end{smallmatrix}\right)_{x+1} |2;n \rangle_{\rm norm},
 \\
 J_L^{(2)}(x) 
&=
w(01 | 10) \langle 2;n| \left(\begin{smallmatrix} 1&& \\ &0& \\ &&0 \end{smallmatrix}\right)_{x} \left(\begin{smallmatrix} 0&& \\ &1& \\ &&0 \end{smallmatrix}\right)_{x+1} |2;n \rangle_{\rm norm}
+ w(02 | 11) \langle 2;n| \left(\begin{smallmatrix} 1&& \\ &0& \\ &&0 \end{smallmatrix}\right)_{x} \left(\begin{smallmatrix} 0&& \\ &0& \\ &&1 \end{smallmatrix}\right)_{x+1} |2;n \rangle_{\rm norm} \\
&+ w(11 | 20) \langle 2;n| \left(\begin{smallmatrix} 0&& \\ &1& \\ &&0 \end{smallmatrix}\right)_{x} \left(\begin{smallmatrix} 0&& \\ &1& \\ &&0 \end{smallmatrix}\right)_{x+1} |2;n \rangle_{\rm norm}
+ w(12 | 21) \langle 2;n| \left(\begin{smallmatrix} 0&& \\ &1& \\ &&0 \end{smallmatrix}\right)_{x} \left(\begin{smallmatrix} 0&& \\ &0& \\ &&1 \end{smallmatrix}\right)_{x+1} |2;n \rangle_{\rm norm}
\\
&+ 2 w(02 | 20) \langle 2;n| \left(\begin{smallmatrix} 1&& \\ &0& \\ &&0 \end{smallmatrix}\right)_{x} \left(\begin{smallmatrix} 0&& \\ &0& \\ &&1 \end{smallmatrix}\right)_{x+1} |2;n \rangle_{\rm norm}. 
\end{split}
\end{equation}
Here $\beta$ is chosen to satisfy the conditions given by (\ref{cond_comb}). The coefficients $w(ab | a'b')$ are transition rates obtained in the update matrix (\ref{3state-markov}):
\begin{equation} \label{hopping_rates}
\begin{split}
 &w(10 | 01) = \tfrac{q^2}{q + q^{-1}},
 \hspace{29mm}
 w(20 | 11) = q + \beta q^3 (q+q^{-1}), \\
 &w(11 | 02) = \tfrac{q^3 + \beta q (q+q^{-1})}{(q+q^{-1})^2}, 
 \hspace{19mm}
 w(21 | 12) = \tfrac{q^2}{q+q^{-1}}, \\
 &w(20 | 02) = -\beta q^4, 
 \hspace{29mm}
 w(01 | 10) = \tfrac{1}{q^2 (q+q^{-1})}, \\
 &w(02 | 11) = q^{-1} + \beta q^{-3} (q+q^{-1}),
 \qquad
 w(11 | 20) = \tfrac{1 + \beta q^2(q+q^{-1})}{q^3 ( q+q^{-1})^2}, \\
 &w(12 | 21) = \tfrac{1}{q^2(q+q^{-1})}, 
 \hspace{24mm}
 w(02 | 20) = -\beta q^{-4}. 
\end{split}
\end{equation}
%
%\begin{figure}
%\begin{center}
% \includegraphics[scale=0.65]{currents2}
% \caption{The definition of right-moving currents on the three-state ASEP and their matrix representations. } \label{fig:currents2}
%\end{center}
%\end{figure}
%
Computation is cumbersome but straightforward. Using formulae (\ref{n-point}), (\ref{Spm-1}), and (\ref{Spm-2}) and a map onto the fundamental representations, one obtains that currents of the three-state ASEP (\ref{currents4}) are expressed by particle-counting operators (Appendix~7). Substituting (\ref{currents3}) into (\ref{currents4}), one obtains 
\begin{equation}
 J(x) = J_R(x) - J_L(x) = 0. 
\end{equation}
Instead of giving explicit forms, we show the plots of right-moving currents and left-moving currents (Fig.~\ref{fig:current}). One obtains non-zero right-moving currents, which are compensated by non-zero left-moving ones, at the surface of the high-density domain, as can be easily expected from particle-density profiles (Fig.~\ref{fig:density}); In the high-density domain, particles are ``frozen'' since each site does not admit more than $\ell$ particles. 
\begin{figure}
\begin{center}
 \includegraphics[scale=0.65]{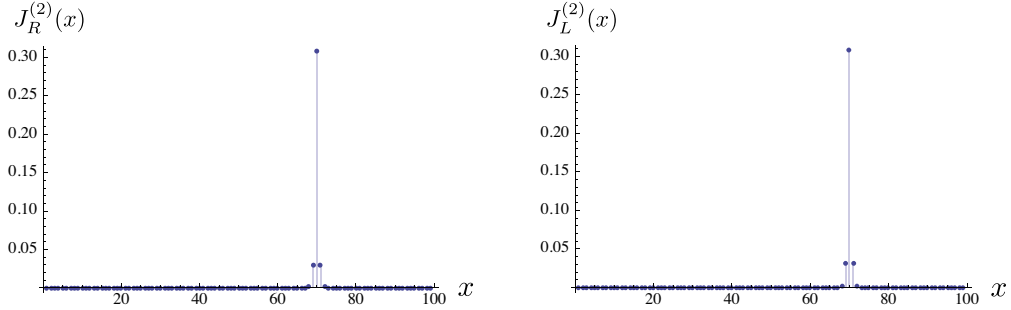}
 \caption{Right-moving and left-moving currents in a steady state of the three-state ASEP are plotted for a system with size $N=100$ with $n=60$ particles. Here we chose a free parameter $\beta$ as $\beta=1/4$ and $q$ as $q=2$. Comparing this plot with Fig.~\ref{fig:density}, one obtains that each current exists only at the surface of the high-density domain with the same amplitudes. }\label{fig:current}
\end{center}
\end{figure}

%%%%%%%%%%%%%%%%%%%%%%%%%%%%%%%%%%%%%%%%%%%%%%%%%%%%%%%%%%%%%%%%%%
%%%%%%%%%%%%%%%%%%%%%%%%%%%%%%%%%%%%%%%%%%%%%%%%%%%%%%%%%%%%%%%%%%%
\section{Large-volume limit}
Representative behaviors of particle densities and currents obtained in Fig.~\ref{fig:density} and \ref{fig:current} are understood from asymptotic analysis of the exact expressions of density profiles and currents. In this section, we analyze these physical quantities in the large-volume limit under the existence of a large enough number of particles. 

%%%%%%%%%%%%%%%%%%%%%%%%%%%%%%%%%%%%%%%%%%%%%%%%%%%%%%%%%%%%%%%%%%%
\subsection{Asymptotics of density profiles}
Since there are no particles coming in and going out, we expect as typical behaviors of particle-density profiles that a larger hopping rate to the right than to the left $q > 1$ brings a high-density domain for $x > x_H$. At the same time, it is expected that a zero-density domain exists for $x < x_L$. As we obtain below, the threshold values $x_{H,L}$ are determined from a decay length. 
\begin{proposition}
Particle-density profiles of the $n$-particle steady state asymptotically behave as 
\begin{align}
 &\rho^{(\ell)}_n(N - \tfrac{n-1}{\ell} - r) = q^{-2\ell r} + \mathcal{O}(q^{-2\ell r}) \sim \exp [-r/\xi] \label{rising*}\\
 &\rho^{(\ell)}_n(N - \tfrac{n+1}{\ell} + 1 + r) = \ell - q^{-2\ell r} + \mathcal{O}(q^{-2\ell r}) \sim \ell - \exp [-r/\xi], \label{falling*} 
\end{align}
for large enough $N$ and $n$ with a decay length given by 
\begin{equation} \label{decay}
 \xi = \frac{1}{2\ell \ln q}. 
\end{equation}
\end{proposition}
\begin{proof}
Remind that the particle-density profile in the $n$-particle steady state of the $(\ell + 1)$-state ASEP is expressed by the formula (\ref{density*}). Since $q$-binomials asymptotically behave as 
\begin{equation}
 \begin{split}
  \begin{bmatrix} \ell N \\ n \end{bmatrix}^{-1} \begin{bmatrix} \ell N \\ k \end{bmatrix}
  &\sim 
  q^{-\frac{1}{2} (\ell N-n)(\ell N-n+1)} 
  q^{-\frac{1}{2} n(n+1)}
  q^{\frac{1}{2} (\ell N-k)(\ell N-k+1)}
  q^{\frac{1}{2} k(k+1)} \\
  &\sim
  q^{(n-k) (n+k-\ell N)}, 
 \end{split}
\end{equation}
density profiles show asymptotic behaviors for $q>1$ as
\begin{equation} \label{density-asym}
\begin{split}
 \rho^{(\ell)}_n(x) 
  &\sim
  \sum_{j=1}^{\ell} \sum_{k=0}^{n-1}
  (-1)^{n-k+1}
  q^{-(n-k) (\ell N+1-2 (\ell (x-1)+j)-n-k+\ell N)} \\
 &= 
  \sum_{j=1}^{\ell} \sum_{k=0}^{n-1}
  (-1)^{n-k+1}
  q^{f(k)}, 
\end{split}
\end{equation}
where 
\begin{equation}
 f(k) = -(n-k) (\ell N+1-2 (\ell (x-1)+j)-n-k+\ell N). 
\end{equation}

Thus, asymptotic behavior of particle density is governed by the maximum value of $f(k)$. The function $f(k)$ is a quadratic function whose zeros are located at $k_1^* = n$ and $k_2^* = 2\ell N - 2(\ell (x-1)+j) - n  + 1$ (Fig.~\ref{fig:quadratic}). 
\begin{figure}
\begin{center}
 \includegraphics[scale=0.65]{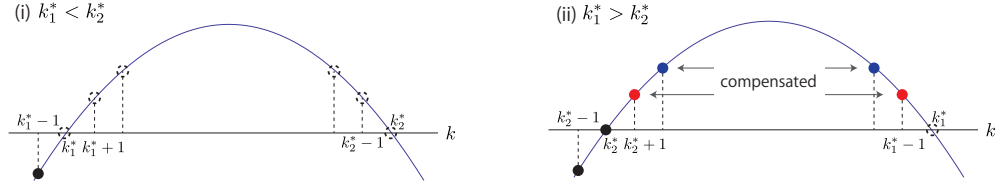}
 \caption{Behaviors of a function $f(k)$ in the case (i) and (ii). While taking the values depicted by filled circles, $k$ is not allowed to take values at dotted circles. The red circles and the blue ones are compensated each other accompanied by opposite signs as is obtained in particle density. } \label{fig:quadratic}
\end{center}
\end{figure}
Then the following two cases are to be considered: 
\begin{enumerate}
 \item[(i)] $k_1^* < k_2^*$ iff $~x < N -\frac{j+n}{\ell} + \frac{1}{2\ell} + 1$ 
 \item[(ii)] $k_1^* > k_2^*$ iff $~x > N -\frac{j+n}{\ell} + \frac{1}{2\ell} + 1$ 
\end{enumerate}
Taking into account that the transition from the zero-particle domain to the high-density domain is assumed to occur around $x \sim \lfloor N - \frac{n}{\ell} \rfloor$, the exponent of the rising edge is analyzed from the case (i), while the falling edge from (ii). 

In the case (i), the maximum value of $f(k)$ is given by $k=k_1^*-1$, since the summation is taken over $k \in \{0, \dots, n-1 = k_1^*-1\}$. The maximum exponent is then derived as 
\begin{equation}
 f(k_1^*-1) = -(k_2^* - k_1^* + 1), 
  %-((2\ell N + 1 - 2j - 2n) - 2\ell x + 1), 
\end{equation}
which takes an negative value since $k_1^* < k_2^*$ in the case (i). Thus, we have the rising exponent as 
\begin{equation} \label{rising}
\begin{split}
 \rho^{(\ell)}_n(x) &\sim \sum_{j=1}^{\ell} q^{-(k_2^* - k_1^* + 1)} \\
 %q^{-(2\ell N - 2n + 2 - 2\ell x - 2j)} \\
 &\sim q^{-(2\ell N - 2\ell x - 2n + 2)}. 
\end{split}
\end{equation}
On the other hand, in the case (ii), the maximum value of $f(k)$ is given by $k = \lfloor \frac{k_1^* + k_2^*}{2} \rfloor$
%$k = \ell N - (\ell x+j)$ 
and $k = \lceil \frac{k_1^* + k_2^*}{2} \rceil$. 
%$k = \ell N - (\ell x+j) + 1$. 
Nevertheless, $q^{f(\lfloor \frac{k_1^* + k_2^*}{2} \rfloor)}$ and $q^{f(\lceil \frac{k_1^* + k_2^*}{2} \rceil)}$ appear in (\ref{density-asym}) accompanied by different signs, and thus these two terms compensate each other. The same happens for $k \in \{k_2^* + 1, \dots, k_1^* - 1\}$ due to 
\begin{equation}
 f(k_2^* + k') = f(k_1^* - k'), \qquad k' \in \{1, \dots, \lfloor \tfrac{k_1^* + k_2^*}{2} \rfloor\}. 
\end{equation}
Thus, the maximum exponent is given by $k = k_2^*$, 
%$k=2\ell N - 2(\ell x + j) - n$,
{\it i.e.} $f(k_2^*) = 0$. The second leading term is then given by $k = k_2^* - 1$, and we obtain the falling exponent as 
\begin{equation} \label{falling}
\begin{split}
 \rho^{(\ell)}_n(x) &\sim 
 \sum_{j=1}^{\ell} \left[
 1 - q^{-(n - k_2^* + 1)}
 %1 - q^{-(-2\ell N + 2n + 2\ell x + 2j)} 
 \right] \\
 &\sim 
 \ell - q^{-(-2\ell N + 2\ell(x-1) + 2n + 2)}. 
\end{split}
\end{equation}
From (\ref{rising}) and (\ref{falling}), asymptotic behaviors are obtained for particle-density profiles, respectively as in (\ref{rising*}) and (\ref{falling*}). 
\end{proof}

Asymptotic behaviors (\ref{rising*}) and (\ref{falling*}) imply that particle density exponentially decays for $x < x_L$ with $x_L = N - \frac{n+1}{\ell} + 1$, while one obtains the high-density domain for $x > x_H$, {\it i.e.} $x_H = N - \frac{n-1}{\ell}$. Thus we conclude that the transition domain antiproportionally shrinks with respect to the number of states. At the same time, the decay length (\ref{decay}) also indicates that the multi-state system shows faster decay in density profiles than the two-state ASEP. 

%%%%%%%%%%%%%%%%%%%%%%%%%%%%%%%%%%%%%%%%%%%%%%%%%%%%%%%%%%%%%%%%%%
%%%%%%%%%%%%%%%%%%%%%%%%%%%%%%%%%%%%%%%%%%%%%%%%%%%%%%%%%%%%%%%%%%
\section{Conclusions}
In this paper, we have constructed the multi-state asymmetric simple exclusion processes based on the fusion procedure of the Temperley-Lieb algebra, satisfied by the Markov matrices. Motivated by higher-spin extension of integrable quantum spin chains, we have discussed a new family of integrable stochastic processes. As is in the case of the two-state ASEP, the multi-state processes proposed in this paper have corresponding quantum spin chains, some of whose wave functions are exactly discussed by the Bethe ansatz method. This fact implies that the multi-state ASEP also admits exact calculation of physical quantities through the Bethe ansatz method. 

Existence of the similarity transformation which makes the Markov matrix to satisfy probability conservation is significant, since otherwise the higher-dimensional Temperley-Lieb generators do not describe stochastic processes. It is also important fact that there exists a linear combination of the fused generators which satisfies the positivity condition. Although we only gave conjectures as the conditions satisfied by the coefficient parameters for an arbitrary $(\ell+1)$-state case, the existence of the $(\ell+1)$-state ASEP is ensured since at least one generator $M_{i,i+1}^{(\ell;1)}$ has only positive off-diagonal elements. 
In the totally asymmetric limit, we gave physical interpretation to the coefficient parameters as friction among particles. Each coefficient is independently controllable as long as it satisfies the positivity condition. Especially in the TASEP limit, each coefficient controls each type of hoppings, {\it i.e.} the parameter $b_r^{(\ell)}$ controls the hopping rate of $r$ particles together. 
The controllability of these parameters and the solvability of the model allows us to used the multi-state ASEP in the discussion of various one-dimensional stochastic processes. 
We hope that the multi-state ASEP shed the light on the exact analysis of fields, such as granular matter, which have been far from analytical treatment. 

Exact analysis of the multi-state ASEP is possible due to beautiful algebraic structure of the Temperley-Lieb algebra, which still holds for the extended model. Based on this algebraic structure, we computed particle-density profiles and particle currents on a steady state. Then a characteristic feature of the multi-state process has been obtained in the decay length which strongly depends on the number of states of the system. The decay length was defined from the rising and falling exponent of particle-density profiles obtained from asymptotic analysis. 

Although we focused on the closed model on which the Temperley-Lieb algebraic structure holds for the whole system, it is more interesting to consider general boundary conditions. The open system would be solved via the matrix product ansatz, although we did not find the matrix product steady state~\cite{bib:A13}. 
Usually, matrix product states are constructed based on transition rules of the model and thus requires the knowledge of explicit forms of an update matrix. Although we gave explicit forms of the update matrices for arbitrary $\ell$, the expressions are too cumbersome and the relations among the fused Temperley-Lieb generators are still unclear. Recently, the $T$-systems and the $Y$-systems of the fused generators were found~\cite{bib:MPR14}, which might help to understand algebraic structure of the fused Temperley-Lieb generators. 
After these problems are resolved, we hope that more detailed analysis including general boundary cases would become possible. 

Another interesting open problem is how the multi-state TASEPs relate with combinatorial problems. It has been shown that current fluctuations of TASEP with the step initial condition, which are also considered as fluctuations of the surface growth model, obey the Tracy--Widom (TW) distribution~\cite{bib:S05, bib:FS05, bib:DF06, bib:S07, bib:D07, bib:IS07, bib:BFPS06, bib:TW08, bib:S08, bib:BFS08}. This TW distribution appears in various context of probability theory such as distribution of the longest increasing subsequences~\cite{bib:AD99, bib:BDJ99, bib:B01, bib:GTW01, bib:CG04, bib:HLM06, bib:HL09, bib:BH10, bib:HR10, bib:HT12}. 
Dynamics of the two-state TASEP has been analyzed by using combinatorics, symmetric polynomials, and the random matrix theory, after mapped to the growth of Young diagrams \cite{bib:J00}. We hope that the similar correspondence between multi-state TASEPs and combinatorial objects will be unveiled soon. 
All the above future works are possible only when the process is integrable. Therefore, we regard out work as an initial but important step toward investigating a large class of non-equilibrium systems and their mathematical aspects. 

%%%%%%%%%%%%%%%%%%%%%%%%%%%%%%%%%%%%%%%%%%%%%%%%%%%%%%%%%%%%%%%%
%%%%%%%%%%%%%%%%%%%%%%%%%%%%%%%%%%%%%%%%%%%%%%%%%%%%%%%%%%%%%%%%

%
\begin{acknowledgement}
The author is grateful to C. Arita, K. Mallick and K. Motegi for helpful discussions. Especially, C.M. would like to thank C. Arita, N. Crampe, E. Ragoucy, and M. Vanicat for valuable comments for improvement of this paper. The author also would like to thank N. Demni for suggesting interesting future works related to combinatorial problems. 

This work is partially supported by the Aihara Project, the FIRST program from JSPS, initiated by CSTP. 
\end{acknowledgement}
%
% appendix
%%%%%%%%%%%%%%%%%%%%%%%%%%%%%%%%%%%%%%%%%%%%%%%%
%%%%%%%%%%%%%%%%%%%%%%%%%%%%%%%%%%%%%%%%%%%%%%%%
%%%%%%%%%%%%%%%%%%%%%%%%%%%%%%%%%%%%%%%%%%%%%%%%
%%%%%%%%%%%%%%%%%%%%%%%%%%%%%%%%%%%%%%%%%%%%%%%
%%%%%%%%%%%%%%%%%%%%%%%%%%%%%%%%%%%%%%%%%%%%%%%
\section*{Appendix~1: The $U_q(sl_2)$ algebra} \label{sec:uqsl2}
\addcontentsline{toc}{section}{Appendix}
%\section{The $U_q(sl_2)$ algebra} \label{sec:uqsl2}
The $U_q(sl_2)$ algebra is generated by $S^{\pm}$ and $q^{S^z}$ which satisfy the following commutation relations: 
\begin{equation}
 q^{S^z} S^{\pm} q^{-S^z} = q^{\pm} S^{\pm}, 
  \qquad
  [S^+,\, S^-] = \frac{q^{2S^z} - q^{-2S^z}}{q - q^{-1}}. 
\end{equation}
As the Hopf algebra, the following comultiplication holds for the $U_q(sl_2)$ generators $X \in \{S^{\pm}, q^{S^z}\}$: 
\begin{equation}
 ({\bm 1} \otimes \Delta) \circ \Delta (X) = (\Delta \otimes {\bm 1}) \circ \Delta (X). 
\end{equation}
By choosing $\Delta(S^{\pm}) = S^{\pm} \otimes q^{-S^z} + q^{S^z} \otimes S^{\pm}$ and defining $\Delta^{(N)}$ by 
\begin{equation}
\begin{split}
 \Delta^{(N)} &= (\underbrace{{\bm 1} \circ \cdots \circ {\bm 1}}_{N-2} \circ \Delta) \cdots ({\bm 1} \circ \Delta) \Delta, \\
 &= \cdots \\
 &= (\Delta \circ \underbrace{{\bm 1} \circ \cdots \circ {\bm 1}}_{N-2}) \cdots (\Delta \circ {\bm 1}) \Delta, 
\end{split} 
\end{equation}
we have the spacially extended generators: 
\begin{align}
 &\Delta^{(N)}(S^{\pm}) = \sum_{j=1}^N q^{S_1^z} \otimes \cdots \otimes q^{S_{j-1}^z} \otimes S^{\pm}_j \otimes q^{-S_{j+1}^z} \otimes \cdots \otimes q^{-S_N^z}, 
\\
 &\Delta^{(N)}(q^{S^z}) = q^{S_1^z} \otimes \cdots \otimes q^{S_N^z}. 
\end{align}
%%%%%%%%%%%%%%%%%%%%%%%%%%%%%%%%%%%%%%%%%%%%%%%
%%%%%%%%%%%%%%%%%%%%%%%%%%%%%%%%%%%%%%%%%%%%%%%
\section*{Appendix~2: Graphical representations of the TL algebra} \label{sec:TLalgebra}
The basis of the TL algebra is known to be described by the link patterns. The link pattern is made by connecting two distinct sites with non-crossing arches. Each of sites is identified with a space of the Temperley-Lieb algebra. Link patterns with different shapes are orthogonal to each other, which form the basis of the TL algebra. 
On this basis, the identity operator just map the original spaces to themselves which is graphically represented as in the left figure of Fig.~\ref{fig:TL-2dim}. On the other hand, the TL generator mixes two spaces and then its graphical representation is given by the right one of Fig.~\ref{fig:TL-2dim}. Then the TL algebraic relations (\ref{TL-relations}) are graphically given by Fig.~\ref{fig:TLrel}. 
\begin{figure}
\begin{center}
 \includegraphics[scale=0.45]{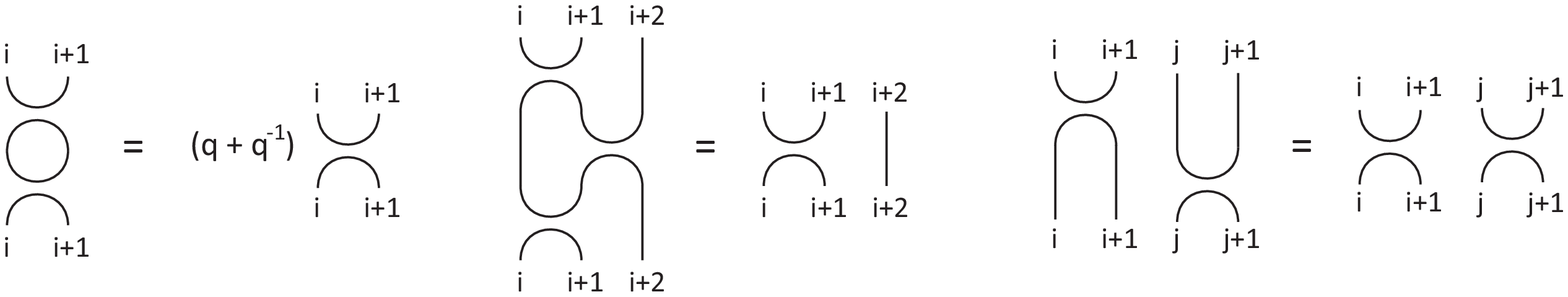}
 \caption{Graphical representations of the algebraic relations (\ref{TL-relations}). If a circle shows up, a weight $(q+q^{-1})$ is added.} \label{fig:TLrel}
\end{center}
\end{figure}

As an illustration, we show the action of $e_2$ on the basis of link patterns in the case of $N=6$ (Fig.~\ref{fig:link_patterns}). 
\begin{figure}
\begin{center}
 \includegraphics[scale=0.5]{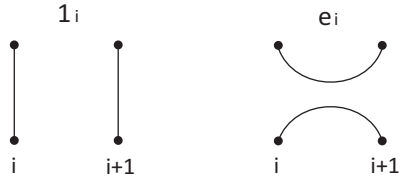}
 \caption{Graphical representations of the identity operator and the TL generator.} \label{fig:TL-2dim}
\end{center}
\end{figure}
\begin{figure}
\begin{center}
 \includegraphics[scale=0.5]{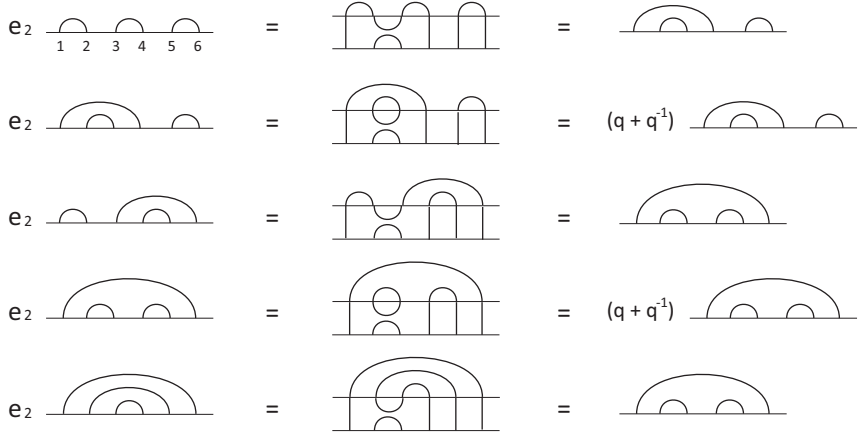}
 \caption{The action of $e_2$ on the basis of link patterns with $N=6$. Each shape can be topologically deformed. A weight $(q+q^{-1})$ is added coming from the first relation of (\ref{TL-relations}).} \label{fig:link_patterns}
\end{center}
\end{figure}

Taking into consideration of the definitions (\ref{TL-3dim1}) and (\ref{TL-3dim2}), it is now naturally obtained that the generators $e_i^{(2;1)}$ and $e_i^{(2;2)}$ are graphically represented as in Fig.~\ref{fig:TL-3dim}. The projection operator $Y^{(2)}$ is denoted by a red circle. 
\begin{figure}
\begin{center}
 \includegraphics[scale=0.5]{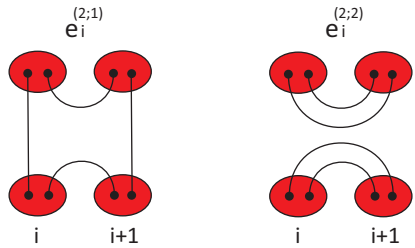}
 \caption{The graphical representations of the three-dimensional fused TL generators $e_i^{(2;1)}$ and $e_i^{(2;2)}$. The projection operator $Y^{(2)}$ is denoted by a red circle.} \label{fig:TL-3dim}
\end{center}
\end{figure}
%

%%%%%%%%%%%%%%%%%%%%%%%%%%%%%%%%%%%%%%%%%%%%%%%
%%%%%%%%%%%%%%%%%%%%%%%%%%%%%%%%%%%%%%%%%%%%%%%
\section*{Appendix~3: $SO(3)$ Birman-Murakami-Wenzl algebra} \label{sec:BMWalgebra}
Here we give algebraic relations between $e_i^{(2;1)}$ and $e_i^{(2;2)}$: 
\begin{align}
 &(e_i^{(2;2)})^2 = (q^2 + 1 + q^{-2}) e_i^{(2;2)}, \\
 &(e_i^{(2;1)})^2 = \frac{q^2 + q^{-2}}{q + q^{-1}} e_i^{(2;1)} + \frac{1}{q^2 + q^{-2}} e_i^{(2;2)}, \\
 &e_i^{(2;1)} e_i^{(2;2)} = \frac{q^2 + 1 + q^{-2}}{q + q^{-1}} e_i^{(2;2)},  
\end{align}
and 
\begin{align}
 &e_i^{(2;2)} e_{i+1}^{(2;2)} e_i^{(2;2)} = e_i^{(2;2)}, \\
 &e_i^{(2;1)} e_{i+1}^{(2;2)} e_i^{(2;1)} = e_{i+1}^{(2;1)} e_i^{(2;2)} e_{i+1}^{(2;1)}, \\
 &e_i^{(2;2)} e_{i+1}^{(2;1)} e_i^{(2;2)} = \tfrac{q^2 + 1 + q^{-2}}{q + q^{-1}} e_i^{(2;2)}, \\
 &e_i^{(2;1)} e_{i+1}^{(2;2)} e_i^{(2;2)} = e_{i+1}^{(2;1) e_i^{(2;2)}}, \\
 &e_i^{(2;1)} e_{i+1}^{(2;1)} e_i^{(2;2)} = \tfrac{q^2 + q^{-2}}{q + q^{-1}} e_{i+1}^{(2;1)} e_i^{(2;2)} + \tfrac{1}{(q + q^{-1})^2} e_i^{(2;2)}, \\
 &e_i^{(2;1)} e_{i+1}^{(2;1)} e_i^{(2;1)} - e_{i+1}^{(2;1)} e_i^{(2;1)} e_{i+1}^{(2;1)} 
 = \tfrac{1}{(q + q^{-1})^2} e_{i+1}^{(2;1)} e_i^{(2;2)} + \tfrac{1}{(q + q^{-1})^2} e_i^{(2;2)} e_{i+1}^{(2;1)} \nonumber \\
 &\hspace{31mm}+ \tfrac{1}{(q + q^{-1})^2} e_i^{(2;1)} - \tfrac{1}{(q + q^{-1})^3} e_i^{(2;2)} - \tfrac{1}{(q + q^{-2})^2} e_i^{(2;1)} e_{i+1}^{(2;2)} \nonumber \\
 &\hspace{31mm}- \tfrac{1}{(q + q^{-1})^2} e_{i+1}^{(2;2)} e_i^{(2;1)} - \tfrac{1}{(q + q^{-1})^3} e_{i+1}^{(2;2)} + \tfrac{1}{(q + q^{-1})^2} e_{i+1}^{(2;1)}. 
\end{align}
These are known as the $SO(3)$ BMW algebra. This kind of algebraic relations have not been found yet for the fused TL generators with $\ell \geq 3$. 

%%%%%%%%%%%%%%%%%%%%%%%%%%%%%%%%%%%%%%%%%%%%%%%
%%%%%%%%%%%%%%%%%%%%%%%%%%%%%%%%%%%%%%%%%%%%%%%
\section*{Appendix~4: Matrix elements of the $e_i^{(\ell;r)}$ generator} \label{sec:fusedTL}
\addcontentsline{toc}{section}{Appendix}
In this appendix, we compute explicit matrix elements of the $e_i^{(\ell;r)}$ generator. First of all, we give $e_i^{(\ell;r)}$ in terms of two-dimensional representations. Let us introduce the following vector notations: 
\begin{equation} \label{double-ket}
\begin{split}
% &\langle\langle m_1,\dots,m_r| = \langle 0\, \dots \,\underset{m_1}{1}\, \dots\, \underset{m_r}{1}\. \dots\, 0|, \\
 &|m_1,\dots,m_r \rangle\rangle = |\underset{1}{0}\, \dots \, \underset{m_1}{1}\, \underset{\dots}{\dots}\, \underset{m_r}{1}\, \dots\, \underset{2r}{0} \rangle, 
\end{split}
\end{equation}
with the definitions (\ref{state-def1}) and (\ref{2state-basis}). 
Then, $e_i^{(\ell;r)}$ is expressed as 
\begin{equation} \label{TL-2dim}
\begin{split}
 e_i^{(\ell;r)} = &\underbrace{\bm{1} \otimes \cdots \otimes \bm{1}}_{\ell - r} \otimes
 \sum_{1 \leq m_1 < \dots < m_r \leq 2r \atop 1 \leq n_1 < \dots < n_r \leq 2r} 
 q^{\sum_{k=1}^r \Theta(\ell-m_k)-\frac{1}{2}}
 q^{\sum_{k=1}^r \Theta(\ell-n_k)-\frac{1}{2}}
 (-1)^{{\rm card}(m_j \neq n_k)}
 \\
 &\cdot \prod_{k=1}^r (1 - \delta_{m_{r-k+1}, 2r-m_k+1})\,
 |m_1 \dots m_r \rangle\rangle \langle\langle n_1 \dots n_r|
 \otimes \underbrace{\bm{1} \otimes \cdots \otimes \bm{1}}_{\ell - r}, 
\end{split}
\end{equation}
where 
\begin{equation}
 \delta_{ab} = 
  \begin{cases}
   1 & a = b \\
   0 & a\neq b. 
  \end{cases}
\end{equation}
The $(\ell+1)$-dimensional representation of this generator is obtained by using the two-dimensional vector basis of the $U_q(sl_2)$ algebra. Since the $s$-particle state is expressed in terms of the notation (\ref{double-ket}) by  
\begin{align}
 &|s \rangle_{\rm norm} = \begin{bmatrix} \ell \\ s \end{bmatrix}^{-1}
 q^{-\frac{s}{2} (\ell+1)} \sum_{1 \leq n_1 < \dots < n_s \leq \ell} q^{\sum_{k=1}^s n_k} |n_1 \dots n_s \rangle\rangle, \\
 &\langle s| = 
 q^{-\frac{s}{2} (\ell+1)} \sum_{1 \leq n_1 < \dots < n_s \leq \ell} q^{\sum_{k=1}^s n_k} \langle\langle n_1 \dots n_s|, 
\end{align}
we have
\begin{equation} \label{TL-ldim}
\begin{split}
 &\langle r_L| \otimes \langle s_L|\; e_i^{(\ell;r)}\; |r_R \rangle \otimes |s_R \rangle
 \\
  &= 
 \begin{bmatrix} \ell \\ r_R \end{bmatrix}^{-1}
 \begin{bmatrix} \ell \\ s_R \end{bmatrix}^{-1}
 q^{-\frac{1}{2}(r_L + s_L + r_R + s_R) (\ell+1)}
  \sum_{1 \leq m_1<\cdots<m_{r_L} \leq \ell \atop {1 \leq n_1<\cdots<n_{s_L} \leq \ell \atop {1 \leq m'_1<\cdots<m'_{r_R} \leq \ell \atop 1 \leq n'_1<\cdots<n'_{s_R} \leq \ell}}}
 q^{\sum_{k=1}^{r_L} m_k} q^{\sum_{j=1}^{s_L} n_j} q^{\sum_{k'=1}^{r_R} m'_{k'}} q^{\sum_{j'=1}^{s_R} n'_{j'}}
 \\
 &\hspace{5mm}\cdot
  \langle\langle m_1 \dots m_{r_L}| \otimes \langle\langle n_1 \dots n_{s_L}|\
 e_i^{(\ell;r)}
 \;|m_1' \dots m_{r_R}' \rangle\rangle \otimes |n_1' \dots n_{s_R}' \rangle\rangle. 
\end{split}
\end{equation}
Here we abbreviated the subscription ``norm''. 
If one sets $m_j$ and $n_j$ in (\ref{TL-2dim}) by 
\begin{align}
 &\bm{\alpha} = \{m_1+\ell-r,\dots,m_r+\ell-r\} = \{a_1,\dots,a_r\}, \\
 &\bm{\beta} = \{n_1+\ell-r,\dots,n_r+\ell-r\} = \{b_1,\dots,b_r\},
\end{align}
(\ref{TL-ldim}) has a non-zero value only when $\bm{\alpha} \in \{m_1, \dots, m_{r_L}, n_1+\ell, \dots, n_{s_L}+\ell\}$ and $\bm{\beta} \in \{m'_1, \dots, m'_{r_R}, n'_1+\ell, \dots, n'_{s_R}+\ell\}$. We also introduce the following notations: 
\begin{align}
 &\bm{\bar{\alpha}} = \{m_1, \dots, m_{r_L}, n_1+\ell, \dots, n_{s_L}+\ell\}\; \backslash\; \bm{\alpha}
 = \{\bar{a}_1, \dots, \bar{a}_{r_L+s_L-r}\}, \\
 &\bm{\bar{\beta}} = \{m'_1, \dots, m'_{r_R}, n'_1+\ell, \dots, n'_{s_R}+\ell\}\; \backslash\; \bm{\beta}
 = \{\bar{b}_1, \dots, \bar{b}_{r_R+s_R-r}\}.
\end{align}
%we obtain  
%\begin{equation} \label{TL-ldim2}
%\begin{split}
% &\langle r_L| \otimes \langle s_L|\; e^{(\ell;r)}\; |r_R \rangle \otimes |s_R \rangle
% \\
% &= 
% \begin{bmatrix} \ell \\ r_L \end{bmatrix}^{-1}
% \begin{bmatrix} \ell \\ s_L \end{bmatrix}^{-1}
% q^{r (2\ell+1)} q^{-\frac{1}{2} (r_L + s_L + r_R + s_R) (\ell+1)} 
% \\
% &\cdot
% \sum_{1 \leq m_1<\cdots<m_{r_L} \leq \ell \atop {1 \leq n_1<\cdots<n_{s_L} \leq \ell \atop {1 \leq m'_1<\cdots<m'_{r_R} \leq \ell \atop 1 \leq n'_1<\cdots<n'_{s_R} \leq \ell}}}
% \sum_{\bm{\alpha} \in \{m_1.\dots,m_{r_L}, n_1+\ell, \dots, n_{s_L}+\ell\} \atop \bm{\beta} \in \{m'_1, \dots, m'_{r_R}, n'_1+\ell, \dots, n'_{s_R}+\ell\}} 
% (-1)^{\sum_{k=1}^r (1 - \delta_{a_k b_k})}
% \\
% &\cdot
% q^{-\sum_{k=1}^r (a_k-\ell+r)} q^{-\sum_{k=1}^r (b_k-\ell+r)}
% q^{\sum_{k=1}^{r_L} m_k} q^{\sum_{j=1}^{s_L} n_j} q^{\sum_{k'=1}^{r_R} m'_{k'}} q^{\sum_{j'=1}^{s_R} n'_{j'}}
% \\
% &\cdot
% \delta_{r_L + s_L, r_R + s_R} \prod_{k=1}^{r_L + s_L - r} \delta_{\bar{a}_k \bar{b}_k}. 
%\end{split}
%\end{equation}
Taking into account the following relations: 
\begin{align}
 &\sum_{k=1}^{r_L} m_k = \sum_{k=1}^{r_L - {\rm card}(a_k \leq \ell)} \bar{a}_k + \sum_{k=1}^{{\rm card}(a_k \leq \ell)} a_k, 
 \\
 &\sum_{j=1}^{s_L} n_j = \sum_{j={\rm card}(a_k \leq \ell)+1}^r (a_j-\ell) + \sum_{j={\rm card}(a_k > \ell)+1}^{r_L + s_L - r} (\bar{a}_j - \ell), 
 \\
 &\sum_{k'=1}^{r_R} m'_{k'} = \sum_{k=1}^{r_R - {\rm card}(b_{j} \leq \ell)} \bar{b}_{k} + \sum_{k=1}^{{\rm card}(b_{j} \leq \ell)+1} b_{k}, 
 \\
 &\sum_{j'=1}^{s_R} n'_{j'} = \sum_{j={\rm card}(b_{j} > \ell)+1}^r (b_{j}-\ell) + \sum_{j={\rm card}(b_{j} > \ell)+1}^{r_R + s_R - r} (\bar{b}_{j} - \ell), 
\end{align}
we have
\begin{equation} \label{TL-ldim3}
\begin{split}
 &\langle r_L| \otimes \langle s_L|\; e_i^{(\ell;r)}\; |r_R \rangle \otimes |s_R \rangle
 \\
 &= 
 \begin{bmatrix}
  \ell \\ r_R
 \end{bmatrix}^{-1}
 \begin{bmatrix}
  \ell \\ s_R
 \end{bmatrix}^{-1}
 \delta_{r_L+s_L, r_R+s_R}
 q^{-\frac{1}{2} (r_L + s_L + r_R + s_R) (\ell+1)}
 \sum_{1\leq m_1<\dots<m_r\leq 2r \atop 1\leq n_1<\dots<n_r\leq 2r} 
 q^{\sum_{k=1}^r \Theta(\ell-m_k)-\frac{1}{2}}
 q^{\sum_{k=1}^r \Theta(\ell-n_k)-\frac{1}{2}}
 \\
 &\cdot
 \sum_{\ell-r+1 \leq a_1 < \dots < a_r \leq \ell+r \atop \ell-r+1 \leq b_1 < \dots < b_r \leq \ell+r}
 (-1)^{{\rm card}(m_j \neq n_k)}
 \prod_{k=1}^r (1 - \delta_{a_{r-k+1}-\ell+r, \ell+r-a_k+1}) (1 - \delta_{b_{r-k+1}-\ell+r, \ell+r-b_k+1})
 \\
 &\cdot
 \sum_{1 \leq \bar{a}_1 < \dots < \bar{a}_{r_L - {\rm card}(a_k \leq \ell)} \leq \ell-r
 \atop{\ell+r+1 \leq \bar{a}_{r_L - {\rm card}(a_k \leq \ell)+1 < \dots < \bar{a}_{r_L+s_L-r} \leq 2\ell}}}
 q^{-\sum_{k=1}^{r_L+s_L-r} \bar{a}_k}
 \sum_{1 \leq \bar{b}_1 < \dots < \bar{b}_{r_R - {\rm card}(b_k \leq \ell)} \leq \ell-r
 \atop{\ell+r+1 \leq \bar{b}_{r_R - {\rm card}(b_k \leq \ell)+1 < \dots < \bar{b}_{r_L+s_L-r} \leq 2\ell}}}
 q^{-\sum_{k=1}^{r_R+s_R-r} \bar{b}_k}
 \\
 &\cdot
 \prod_{k=1}^{r_L+s_L-r} \delta_{\bar{a}_k\bar{b}_k} \cdot
 q^{-\ell\cdot{\rm card}(a_k>\ell)} q^{-\ell\cdot{\rm card}(\bar{a}_k>\ell)} q^{-\ell\cdot{\rm card}(b_k>\ell)} q^{-\ell\cdot{\rm card}(\bar{b}_k>\ell)}. 
\end{split}
\end{equation}
Especially, non-zero elements of $e_i^{(\ell;1)}$ are obtained from explicit calculation as 
\begin{align}
 &\langle r| \otimes \langle s|\; e_i^{(\ell;1)} \;|r \rangle \otimes |s \rangle
 = 
 q^{-r+s+\ell+1} \frac{[r] [\ell-s]}{[\ell]^2}
 +q^{-r+s-\ell-1} \frac{[\ell-r] [s]}{[\ell]^2}, \label{elements-1}
 \\
 &\langle r+1| \otimes \langle s-1|\; e_i^{(\ell;1)} \;|r \rangle \otimes |s \rangle
 = 
 -
 q^{-r+s-1} \frac{[\ell-r] [s]}{[\ell]^2}, \label{elements-2}
 \\
 &\langle r-1| \otimes \langle s+1|\; e_i^{(\ell;1)} \;|r \rangle \otimes |s \rangle
 = 
 -
 q^{-r+s+1} \frac{[\ell-r] [s]}{[\ell]^2}. \label{elements-3}
\end{align}
$r$ and $s$ run from $0$ to $\ell$ for (\ref{elements-1}), while they run from $1$ to $\ell-1$ for (\ref{elements-2}) and (\ref{elements-3}). 
The expression (\ref{TL-ldim3}) naturally leads to the particle-conservation law given by $r_L + s_L = r_R + s_R$.

%%%%%%%%%%%%%%%%%%%%%%%%%%%%%%%%%%%%%%%%%%%%%%%
%%%%%%%%%%%%%%%%%%%%%%%%%%%%%%%%%%%%%%%%%%%%%%%
\section*{Appendix~5: Derivation of the norms} \label{sec:norms}
\addcontentsline{toc}{section}{Appendix}
In derivation of the norms of an $n$-particle steady state of the $(\ell+1)$-state ASEP, we use the commutation relation of the $U_q(sl_2)$ generators: 
\begin{equation}
 [\Delta(S^+),\, \Delta(S^-)] = \frac{\Delta(q^{2S^z}) - \Delta(q^{-2S^z})}{q - q^{-1}}. 
\end{equation}
This relation leads to 
\begin{equation}
\begin{split}
 (\Delta(S^+))^n (\Delta(S^-))^n &= (\Delta(S^+))^{n-1} 
 \left\{ \Delta(S^-) \Delta(S^+) + \frac{\Delta(q^{2S^z}) - \Delta(q^{-2S^z})}{q - q^{-1}} \right\}
 (\Delta(S^-))^{n-1} \\
&= (\Delta(S^+))^{n-1} (\Delta(S^-))^n \Delta(S^+) \\
 &+ (\Delta(S^+))^{n-1} (\Delta(S^-))^{n-1} \sum_{j = 1}^n \frac{q^{-2(n-j)} \Delta(q^{2S^z}) - q^{2(n-j)} \Delta(q^{-2S^z})}{q - q^{-1}}. 
\end{split}
\end{equation}
Using the fact that $\Delta(S^+) |0\rangle = 0$ and $\Delta(q^{\pm 2S^z}) | 0\rangle = q^{\pm N}$, we obtain the following recursion relation: 
\begin{equation}
\begin{split}
 &\langle 0| (\Delta^{(\ell N)}(S^+))^n (\Delta^{(\ell N)}(S^-))^n |0 \rangle \\
 &= \sum_{j = 1}^n \frac{q^{\ell N-2(n-j)} - q^{-\ell N+2(n-j)}}{q - q^{-1}} \langle 0| (\Delta^{(\ell N)}(S^+))^{n-1} (\Delta^{(\ell N)}(S^-))^{n-1} |0 \rangle. 
\end{split}
\end{equation}
Taking into account that the initial condition: 
\begin{equation}
\langle 0| \Delta^{(\ell N)}(S^+) \Delta^{(\ell N)}(S^-) |0 \rangle
  = \frac{q^{\ell N} - q^{-\ell N}}{q - q^{-1}}, 
\end{equation}
we obtain the expression (\ref{norm}) for the norm. 

%%%%%%%%%%%%%%%%%%%%%%%%%%%%%%%%%%%%%%%%%%%%%%%%%%%%%%%%%%%%%%%%%%%%%%%%%%%%%%%%%%%%%%
\section*{Appendix~6: Correlation functions} \label{sec:corr_func}
\addcontentsline{toc}{section}{Appendix}
We introduce a particle-counting operator defined on a two-dimensional vector space: 
\begin{equation} \label{count_op}
n_j = S_j^+ S_j^- =  \begin{pmatrix} 0 & \\ & 1 \end{pmatrix}_j, \qquad
1 - n_j = S_j^- S_j^+ = \begin{pmatrix} 1 & \\ & 0 \end{pmatrix}_j. 
\end{equation}
In the case of the two-state ASEP, important physical quantities such as particle densities and particle currents are expressed by $n_j$. 

For instance, an $l$-point correlation function of the two-state ASEP is written by means of particle-counting operators in the following way: 
\begin{equation}
 \langle n|U n_{x_1} n_{x_2} \dots n_{x_l} U^{-1}|n \rangle_{\rm norm}
  = 
  \langle n|n_{x_1} n_{x_2} \dots n_{x_l}|n \rangle_{\rm norm}. 
\end{equation}
Useful formulae have been obtained in~\cite{bib:SS94}: 
The one-point correlation functions is given by 
\begin{equation} \label{1p-corr}
 \langle n|U^{-1} n_x U|n \rangle =
 \langle n| n_x |n \rangle =
   \begin{bmatrix} N \\ n \end{bmatrix}^{-1}
   \sum_{k=0}^{n-1} (-1)^{n-k+1} q^{-(n-k) (N+1-2x)}
    \begin{bmatrix} N \\ k \end{bmatrix}, 
\end{equation}
which was derived using the following relations: 
\begin{equation} \label{creation}
S^+_x |n\rangle = q^{(-N-1+2x)/2} (1-n_x) |n-1 \rangle, \qquad
\langle n| S^-_x = q^{(-N-1+2x)/2} \langle n-1| (1-n_x). 
\end{equation}
The relations (\ref{creation}) lead to a recursion relation for an $l$-point correlation function with respect to $n$: 
\begin{equation} \label{l-corr}
\langle n | n_{x_1} \dots n_{x_l} | n \rangle_{\rm norm} 
= \frac{[n] q^{- N - 1 + 2x}}{[N - n + 1]} 
	\langle n-1 | n_{x_1} \dots n_{x_l-1} (1 - n_{x_l}) | n-1 \rangle_{\rm norm}. 
\end{equation}

In contrast to the recursion relation (\ref{l-corr}), by which one needs to compute correlation functions in basis of different particle-sectors, we found another recursion relation which does not change the number of particles: 
\begin{proposition} \label{prop:n-point}
An $l$-point function is decomposed into one-point functions: 
\begin{equation} \label{n-point}
\langle n| n_{x_1} n_{x_2} \cdots n_{x_{l}} | n\rangle_{\rm norm}
= \sum_{j=1}^{l} \prod_{k=1 \atop k\neq j}^{\ell} \frac{q^{(x_k - x_j)}}{q^{(x_k - x_j)} - q^{-(x_k - x_j)}} \cdot
	\langle n| n_{x_j} |n \rangle_{\rm norm}. 
\end{equation}
\end{proposition}
\begin{proof}
The proof is given by an induction on $l$. Before starting the proof, let us remark the following lemma:
\begin{lemma}
 A two-point function is decomposed into one-point functions: 
 \begin{equation} \label{2p-lemma}
%\begin{split}
(q^{x_2-x_1}-q^{-(x_2-x_1)}) \cdot \langle n| n_{x_1} n_{x_2} |n \rangle_{\rm norm}
= q^{x_2-x_1} \cdot \langle n| n_{x_1} |n \rangle_{\rm norm} - q^{-(x_2-x_1)} \cdot \langle n| n_{x_2} |n \rangle_{\rm norm}. 
%\end{split}
\end{equation}
\end{lemma}
\begin{proof}
 This lemma can be proved by considering two expressions of the following function; 
\begin{equation}
 \langle n-1| (1 - n_{x_1}) (1 - n_{x_2}) |n-1 \rangle_{\rm norm}. 
\end{equation}
First applying the formula (\ref{l-corr}) to the operator $(1-n_{x_2})$, one obtains 
 \begin{equation} \label{2p-lemma2}
\begin{split}
&\langle n-1| ( 1 - n_{x_1} )( 1 - n_{x_2} ) |n-1 \rangle_{\rm norm} \\
&= \langle n-1| (1 - n_{x_2}) |n-1 \rangle_{\rm norm} - \langle n-1| n_{x_1} (1 - n_{x_2}) |n-1 \rangle_{\rm norm} \\
&= \frac{[N-n+1]}{[n] q^{-N-1+2x_2}} \cdot 
	\langle n| (n_{x_2} - n_{x_1} n_{x_2}) |n \rangle_{\rm norm}, 
\end{split}
\end{equation}
and then applying to the operator $(1-n_{x_1})$, one has 
\begin{equation} \label{2p-lemma3}
\begin{split}
&\langle n-1| ( 1 - n_{x_1} )( 1 - n_{x_2} ) |n-1 \rangle_{\rm norm} \\
&= \langle n-1| (1 - n_{x_1}) |n-1 \rangle_{\rm norm} - \langle n-1| n_{x_2} (1 - n_{x_1}) |n-1 \rangle_{\rm norm} \\
&= \frac{[N-n+1]}{[n] q^{-N-1+2x_1}} \cdot 
	\langle n| (n_{x_1} - n_{x_1} n_{x_2}) |n \rangle_{\rm norm}. 
\end{split}
\end{equation}
From (\ref{2p-lemma2}) and (\ref{2p-lemma3}), decomposition of two-point functions (\ref{2p-lemma}) is obtained. 
\end{proof}
Assume the relation (\ref{n-point}) holds for an $l$-point function. Then an $(l + 1)$-point function is evaluated as
\begin{equation} 
\begin{split}
\langle n| n_{x_1} n_{x_2} \cdots n_{x_{l}} n_{x_{l + 1}} |n \rangle_{\rm norm} 
&= \sum_{j=1}^{l} \prod_{k=1 \atop k\neq j}^{l} \frac{q^{x_k - x_j}}{q^{x_k - x_j} - q^{-(x_k - x_j)}} \cdot 
	\langle n| n_{n_{x_j}} n_{x_{l+1}} |n \rangle_{\rm norm}. 
\end{split}
\end{equation} 
Using the decomposition formula of two-point functions into one-point functions (\ref{2p-lemma}) to the right-hand side, we obtain the relation (\ref{n-point}) for an $(l+1)$-point function. 
\end{proof}

Substituting the expression for the one-point function (\ref{1p-corr}) into (\ref{n-point}), one obtains the expression for the $l$-point function.

\ifx10
Now we show how the $l$-point correlation function can be derived. In principle, the recursion relation (\ref{l-corr}) allows us to compute correlation functions for arbitrary $l$. As the simplest example, we show the derivation of the two-point functions. 

Using the recursion relation (\ref{l-corr}), the two-point function is reduced to a summation of the two-point function and one-point function: 
\begin{equation}
\begin{split}
 \langle n| n_{x_1} n_{x_2} |n \rangle_{\rm norm}
 &= \langle n-1| n_{x_1} |n-1 \rangle_{\rm norm}
 - \langle n-1| n_{x_1} n_{x_2} |n-1 \rangle_{\rm norm}.
\end{split}
\end{equation}
Using the initial condition obtained as 
\begin{equation}
\langle 2| n_x n_{x+1} |2 \rangle_{\rm norm} 
= 
(-1)^N \begin{bmatrix} N \\ n \end{bmatrix}^{-1}
	q^{-n (N+1-2x)}, 
\end{equation}
we have the expression for the two-point function:  
\begin{equation}
\langle n| n_x n_{x+1} |n \rangle_{\rm norm} 
=
\sum_{k=1}^{n-1} \sum_{j=0}^{k-1} (-1)^{n-j} 
	\begin{bmatrix} N \\ n \end{bmatrix}^{-1}
	\begin{bmatrix} N \\ j \end{bmatrix}
	q^{-(n-j) (N+1-2x)}. 
\end{equation}
\fi

%%%%%%%%%%%%%%%%%%%%%%%%%%%%%%%%%%%%%%%%%%%
%%%%%%%%%%%%%%%%%%%%%%%%%%%%%%%%%%%%%%%%%%%
\section*{Appendix~7: Currents in terms of the particle-counting operators} \label{sec:currents}
\addcontentsline{toc}{section}{Appendix}
Here we give useful expressions for the seven types of expectation values in (\ref{currents4}) in terms of particle-counting operators. 
\begin{equation} \label{currents3}
\begin{split}
 &\langle 2;n| \left(\begin{smallmatrix} 0&& \\ &1& \\ &&0 \end{smallmatrix}\right)_{x} \left(\begin{smallmatrix} 1&& \\ &0& \\ &&0 \end{smallmatrix}\right)_{x+1} |2;n \rangle_{\rm norm} \\
 &= -\tfrac{q^{-5} (q+q^{-1})}{(q-q^{-1}) (q^2-q^{-2}) (q^3-q^{-3})} \langle n| n_{2x-1} |n \rangle_{\rm norm}
 - \tfrac{q^{-3} (q+q^{-1})}{(q^{-1}-q) (q-q^{-1}) (q^2-q^{-2})} \langle n| n_{2x} |n \rangle_{\rm norm} \\
 &\hspace{4mm} 
 - \tfrac{q^{-1} (q+q^{-1})}{(q^{-2}-q^2) (q^{-1}-q) (q-q^{-1})} \langle n| n_{2x+1} |n \rangle_{\rm norm}
 - \tfrac{q (q+q^{-1})}{(q^{-3}-q^3) (q^{-2}-q^2) (q^{-1}-q)} \langle n| n_{2x+2} |n \rangle_{\rm norm}, 
\\[2ex] 
 &\langle 2;n| \left(\begin{smallmatrix} 0&& \\ &0& \\ &&1 \end{smallmatrix}\right)_{x} \left(\begin{smallmatrix} 0&& \\ &1& \\ &&0 \end{smallmatrix}\right)_{x+1} |2;n \rangle_{\rm norm} \\
 &= -\tfrac{q (q+q^{-1})}{(q-q^{-1}) (q^2-q^{-2}) (q^3-q^{-3})} \langle n| n_{2x-1} |n \rangle_{\rm norm}
 - \tfrac{q^{-1} (q+q^{-1})}{(q^{-1}-q) (q-q^{-1}) (q^2-q^{-2})} \langle n| n_{2x} |n \rangle_{\rm norm} \\
 &\hspace{4mm} 
 - \tfrac{q^{-3} (q+q^{-1})}{(q^{-2}-q^2) (q^{-1}-q) (q-q^{-1})} \langle n| n_{2x+1} |n \rangle_{\rm norm}
 - \tfrac{q^{-5} (q+q^{-1})}{(q^{-3}-q^3) (q^{-2}-q^2) (q^{-1}-q)} \langle n| n_{2x+2} |n \rangle_{\rm norm},  
\\[2ex]
 &\langle 2;n| \left(\begin{smallmatrix} 0&& \\ &0& \\ &&1 \end{smallmatrix}\right)_{x} \left(\begin{smallmatrix} 1&& \\ &0& \\ &&0 \end{smallmatrix}\right)_{x+1} |2;n \rangle_{\rm norm} \\
 &= -\tfrac{q^{-4}}{(q-q^{-1}) (q^2-q^{-2}) (q^3-q^{-3})} \langle n| n_{2x-1} |n \rangle_{\rm norm}
 - \tfrac{q^{-4}}{(q^{-1}-q) (q-q^{-1}) (q^2-q^{-2})} \langle n| n_{2x} |n \rangle_{\rm norm} \\
 &\hspace{4mm} 
 - \tfrac{q^{-4}}{(q^{-2}-q^2) (q^{-1}-q) (q-q^{-1})} \langle n| n_{2x+1} |n \rangle_{\rm norm}
 - \tfrac{q^{-4}}{(q^{-3}-q^3) (q^{-2}-q^2) (q^{-1}-q)} \langle n| n_{2x+2} |n \rangle_{\rm norm}, 
\\[2ex]
 &\langle 2;n| \left(\begin{smallmatrix} 0&& \\ &1& \\ &&0 \end{smallmatrix}\right)_{x} \left(\begin{smallmatrix} 0&& \\ &1& \\ &&0 \end{smallmatrix}\right)_{x+1} |2;n \rangle_{\rm norm} \\
 &= -\tfrac{(q+q^{-1})^2}{(q-q^{-1}) (q^2-q^{-2}) (q^3-q^{-3})} \langle n| n_{2x-1} |n \rangle_{\rm norm}
 - \tfrac{(q+q^{-1})^2}{(q^{-1}-q) (q-q^{-1}) (q^2-q^{-2})} \langle n| n_{2x} |n \rangle_{\rm norm} \\
 &\hspace{4mm} 
 - \tfrac{(q+q^{-1})^2}{(q^{-2}-q^2) (q^{-1}-q) (q-q^{-1})} \langle n| n_{2x+1} |n \rangle_{\rm norm}
 - \tfrac{(q+q^{-1})^2}{(q^{-3}-q^3) (q^{-2}-q^2) (q^{-1}-q)} \langle n| n_{2x+2} |n \rangle_{\rm norm}, 
\\[2ex]
 &\langle 2;n| \left(\begin{smallmatrix} 1&& \\ &0& \\ &&0 \end{smallmatrix}\right)_{x} \left(\begin{smallmatrix} 0&& \\ &1& \\ &&0 \end{smallmatrix}\right)_{x+1} |2;n \rangle_{\rm norm} \\
 &= -\tfrac{q^{-1} (q+q^{-1})}{(q-q^{-1}) (q^2-q^{-2}) (q^3-q^{-3})} \langle n| n_{2x-1} |n \rangle_{\rm norm}
 - \tfrac{q (q+q^{-1})}{(q^{-1}-q) (q-q^{-1}) (q^2-q^{-2})} \langle n| n_{2x} |n \rangle_{\rm norm} \\
 &\hspace{4mm} 
 - \tfrac{q^3 (q+q^{-1})}{(q^{-2}-q^2) (q^{-1}-q) (q-q^{-1})} \langle n| n_{2x+1} |n \rangle_{\rm norm}
 - \tfrac{q^5 (q+q^{-1})}{(q^{-3}-q^3) (q^{-2}-q^2) (q^{-1}-q)} \langle n| n_{2x+2} |n \rangle_{\rm norm}, 
\\[2ex]
 &\langle 2;n| \left(\begin{smallmatrix} 0&& \\ &1& \\ &&0 \end{smallmatrix}\right)_{x} \left(\begin{smallmatrix} 0&& \\ &0& \\ &&1 \end{smallmatrix}\right)_{x+1} |2;n \rangle_{\rm norm} \\
 &= -\tfrac{q^5 (q+q^{-1})}{(q-q^{-1}) (q^2-q^{-2}) (q^3-q^{-3})} \langle n| n_{2x-1} |n \rangle_{\rm norm}
 - \tfrac{q^3 (q+q^{-1})}{(q^{-1}-q) (q-q^{-1}) (q^2-q^{-2})} \langle n| n_{2x} |n \rangle_{\rm norm} \\
 &\hspace{4mm} 
 - \tfrac{q (q+q^{-1})}{(q^{-2}-q^2) (q^{-1}-q) (q-q^{-1})} \langle n| n_{2x+1} |n \rangle_{\rm norm}
 - \tfrac{q^{-1} (q+q^{-1})}{(q^{-3}-q^3) (q^{-2}-q^2) (q^{-1}-q)} \langle n| n_{2x+2} |n \rangle_{\rm norm}, 
\\[2ex]
 &\langle 2;n| \left(\begin{smallmatrix} 1&& \\ &0& \\ &&0 \end{smallmatrix}\right)_{x} \left(\begin{smallmatrix} 0&& \\ &0& \\ &&1 \end{smallmatrix}\right)_{x+1} |2;n \rangle_{\rm norm} \\
 &= -\tfrac{q^4}{(q-q^{-1}) (q^2-q^{-2}) (q^3-q^{-3})} \langle n| n_{2x-1} |n \rangle_{\rm norm}
 - \tfrac{q^4}{(q^{-1}-q) (q-q^{-1}) (q^2-q^{-2})} \langle n| n_{2x} |n \rangle_{\rm norm} \\
 &\hspace{4mm} 
 - \tfrac{q^4}{(q^{-2}-q^2) (q^{-1}-q) (q-q^{-1})} \langle n| n_{2x+1} |n \rangle_{\rm norm}
 - \tfrac{q^4}{(q^{-3}-q^3) (q^{-2}-q^2) (q^{-1}-q)} \langle n| n_{2x+2} |n \rangle_{\rm norm}. 
\end{split}
\end{equation}

%\input{referenc}
%Bibliography
\bibliographystyle{plain}
\bibliography{reference}
\end{document}